\documentclass[journal,dvipsnames]{IEEEtran}
\usepackage{graphicx,url}
\graphicspath{ {./images/} }
\usepackage{amssymb,amsmath,mathtools}
\usepackage{tikz,pgfplots}
\usepackage{pgfplots}
\usepackage{pgfplotstable}
\usetikzlibrary{pgfplots.groupplots}
\usepgfplotslibrary{colorbrewer}
\pgfplotsset{compat = 1.15, cycle list/Set1-8}
\usetikzlibrary{pgfplots.statistics, pgfplots.colorbrewer}
\usetikzlibrary{pgfplots.groupplots}
\usetikzlibrary{shapes.geometric,backgrounds,patterns, trees}
\usetikzlibrary{3d,decorations.text,shapes.arrows,positioning,fit,backgrounds}
\usetikzlibrary{positioning, decorations.pathmorphing, shapes}
\usetikzlibrary{decorations.pathreplacing}
\usetikzlibrary{shapes.geometric,backgrounds,patterns, trees}
\usetikzlibrary{spy}
\usetikzlibrary{arrows.meta,
                bending,
                intersections,
                quotes,
                shapes.geometric}
              \usetikzlibrary{automata, positioning}
              \usepgfplotslibrary{fillbetween}
\usetikzlibrary{shapes,arrows}
\usetikzlibrary{arrows.meta}
\usetikzlibrary{positioning}
\tikzset{set/.style={draw,circle,inner sep=0pt,align=center}}
\usetikzlibrary{automata, positioning}
  \usetikzlibrary{shapes,shadows}
  \tikzstyle{abstractbox} = [draw=black, fill=white, rectangle,
  inner sep=10pt, style=rounded corners, drop shadow={fill=black,
  opacity=1}]
\tikzstyle{abstracttitle} =[fill=white]
\usetikzlibrary{calc,positioning,shapes.geometric}
\usetikzlibrary{arrows.meta,arrows}
\usetikzlibrary{matrix}

\colorlet{myRed}{red!20}
\tikzset{
  rows/.style 2 args={/utils/temp/.style={row ##1/.append style={nodes={#2}}},
    /utils/temp/.list={#1}},
  columns/.style 2 args={/utils/temp/.style={column ##1/.append style={nodes={#2}}},
    /utils/temp/.list={#1}}}
\usetikzlibrary{backgrounds,calc,shadings,shapes.arrows,shapes.symbols,shadows}
\definecolor{switch}{HTML}{006996}
\pgfkeys{/pgf/.cd,
  parallelepiped offset x/.initial=2mm,
  parallelepiped offset y/.initial=2mm
}
\pgfdeclareshape{parallelepiped}
{
  \inheritsavedanchors[from=rectangle] 
  \inheritanchorborder[from=rectangle]
  \inheritanchor[from=rectangle]{north}
  \inheritanchor[from=rectangle]{north west}
  \inheritanchor[from=rectangle]{north east}
  \inheritanchor[from=rectangle]{center}
  \inheritanchor[from=rectangle]{west}
  \inheritanchor[from=rectangle]{east}
  \inheritanchor[from=rectangle]{mid}
  \inheritanchor[from=rectangle]{mid west}
  \inheritanchor[from=rectangle]{mid east}
  \inheritanchor[from=rectangle]{base}
  \inheritanchor[from=rectangle]{base west}
  \inheritanchor[from=rectangle]{base east}
  \inheritanchor[from=rectangle]{south}
  \inheritanchor[from=rectangle]{south west}
  \inheritanchor[from=rectangle]{south east}
  \backgroundpath{
    \southwest \pgf@xa=\pgf@x \pgf@ya=\pgf@y
    \northeast \pgf@xb=\pgf@x \pgf@yb=\pgf@y
    \pgfmathsetlength\pgfutil@tempdima{\pgfkeysvalueof{/pgf/parallelepiped
      offset x}}
    \pgfmathsetlength\pgfutil@tempdimb{\pgfkeysvalueof{/pgf/parallelepiped
      offset y}}
    \def\ppd@offset{\pgfpoint{\pgfutil@tempdima}{\pgfutil@tempdimb}}
    \pgfpathmoveto{\pgfqpoint{\pgf@xa}{\pgf@ya}}
    \pgfpathlineto{\pgfqpoint{\pgf@xb}{\pgf@ya}}
    \pgfpathlineto{\pgfqpoint{\pgf@xb}{\pgf@yb}}
    \pgfpathlineto{\pgfqpoint{\pgf@xa}{\pgf@yb}}
    \pgfpathclose
    \pgfpathmoveto{\pgfqpoint{\pgf@xb}{\pgf@ya}}
    \pgfpathlineto{\pgfpointadd{\pgfpoint{\pgf@xb}{\pgf@ya}}{\ppd@offset}}
    \pgfpathlineto{\pgfpointadd{\pgfpoint{\pgf@xb}{\pgf@yb}}{\ppd@offset}}
    \pgfpathlineto{\pgfpointadd{\pgfpoint{\pgf@xa}{\pgf@yb}}{\ppd@offset}}
    \pgfpathlineto{\pgfqpoint{\pgf@xa}{\pgf@yb}}
    \pgfpathmoveto{\pgfqpoint{\pgf@xb}{\pgf@yb}}
    \pgfpathlineto{\pgfpointadd{\pgfpoint{\pgf@xb}{\pgf@yb}}{\ppd@offset}}
  }
}

\makeatletter
\tikzset{anchor/.append code=\let\tikz@auto@anchor\relax,
  add font/.code=%
    \expandafter\def\expandafter\tikz@textfont\expandafter{\tikz@textfont#1},
  left delimiter/.style 2 args={append after command={\tikz@delimiter{south east}
    {south west}{every delimiter,every left delimiter,#2}{south}{north}{#1}{.}{\pgf@y}}}}
\tikzstyle{sms} = [rectangle callout, draw,very thick, rounded corners, minimum height=20pt]
\makeatletter
\tikzset{anchor/.append code=\let\tikz@auto@anchor\relax,
  add font/.code=%
    \expandafter\def\expandafter\tikz@textfont\expandafter{\tikz@textfont#1},
  left delimiter/.style 2 args={append after command={\tikz@delimiter{south east}
    {south west}{every delimiter,every left delimiter,#2}{south}{north}{#1}{.}{\pgf@y}}}}
\tikzstyle{sms} = [rectangle callout, draw,very thick, rounded corners, minimum height=20pt]
\usetikzlibrary{positioning,calc}
\tikzstyle{block} = [rectangle, draw,
text width=10.5em, text centered, rounded corners, minimum height=4em]
\tikzstyle{line} = [draw, -latex]
\tikzset{l3 switch/.style={
    parallelepiped,fill=switch, draw=white,
    minimum width=0.75cm,
    minimum height=0.75cm,
    parallelepiped offset x=1.75mm,
    parallelepiped offset y=1.25mm,
    path picture={
      \node[fill=white,
        circle,
        minimum size=6pt,
        inner sep=0pt,
        append after command={
          \pgfextra{
            \foreach \angle in {0,45,...,360}
            \draw[-latex,fill=white] (\tikzlastnode.\angle)--++(\angle:2.25mm);
          }
        }
      ]
       at ([xshift=-0.75mm,yshift=-0.5mm]path picture bounding box.center){};
    }
  },
  ports/.style={
    line width=0.3pt,
    top color=gray!20,
    bottom color=gray!80
  },
  rack switch/.style={
    parallelepiped,fill=white, draw,
    minimum width=1.25cm,
    minimum height=0.25cm,
    parallelepiped offset x=2mm,
    parallelepiped offset y=1.25mm,
    xscale=-1,
    path picture={
      \draw[top color=gray!5,bottom color=gray!40]
      (path picture bounding box.south west) rectangle
      (path picture bounding box.north east);
      \coordinate (A-west) at ([xshift=-0.2cm]path picture bounding box.west);
      \coordinate (A-center) at ($(path picture bounding box.center)!0!(path
        picture bounding box.south)$);
      \foreach \x in {0.275,0.525,0.775}{
        \draw[ports]([yshift=-0.05cm]$(A-west)!\x!(A-center)$)
          rectangle +(0.1,0.05);
        \draw[ports]([yshift=-0.125cm]$(A-west)!\x!(A-center)$)
          rectangle +(0.1,0.05);
       }
      \coordinate (A-east) at (path picture bounding box.east);
      \foreach \x in {0.085,0.21,0.335,0.455,0.635,0.755,0.875,1}{
        \draw[ports]([yshift=-0.1125cm]$(A-east)!\x!(A-center)$)
          rectangle +(0.05,0.1);
      }
    }
  },
  server/.style={
    parallelepiped,
    fill=white, draw,
    minimum width=0.35cm,
    minimum height=0.75cm,
    parallelepiped offset x=3mm,
    parallelepiped offset y=2mm,
    xscale=-1,
    path picture={
      \draw[top color=gray!5,bottom color=gray!40]
      (path picture bounding box.south west) rectangle
      (path picture bounding box.north east);
      \coordinate (A-center) at ($(path picture bounding box.center)!0!(path
        picture bounding box.south)$);
      \coordinate (A-west) at ([xshift=-0.575cm]path picture bounding box.west);
      \draw[ports]([yshift=0.1cm]$(A-west)!0!(A-center)$)
        rectangle +(0.2,0.065);
      \draw[ports]([yshift=0.01cm]$(A-west)!0.085!(A-center)$)
        rectangle +(0.15,0.05);
      \fill[black]([yshift=-0.35cm]$(A-west)!-0.1!(A-center)$)
        rectangle +(0.235,0.0175);
      \fill[black]([yshift=-0.385cm]$(A-west)!-0.1!(A-center)$)
        rectangle +(0.235,0.0175);
      \fill[black]([yshift=-0.42cm]$(A-west)!-0.1!(A-center)$)
        rectangle +(0.235,0.0175);
    }
  },
}

\usetikzlibrary{calc, shadings, shadows, shapes.arrows}
\tikzset{cross/.style={cross out, draw=black, minimum size=2*(#1-\pgflinewidth), inner sep=0pt, outer sep=0pt},
cross/.default={1pt}}
\tikzset{%
  interface/.style={draw, rectangle, rounded corners, font=\LARGE\sffamily},
  ethernet/.style={interface, fill=yellow!50},
  serial/.style={interface, fill=green!70},
  speed/.style={sloped, anchor=south, font=\large\sffamily},
  route/.style={draw, shape=single arrow, single arrow head extend=4mm,
    minimum height=1.7cm, minimum width=3mm, white, fill=switch!20,
    drop shadow={opacity=.8, fill=switch}, font=\tiny}
}

\usepackage{float}
\definecolor{bluetwo}{RGB}{189, 213, 234}
\definecolor{bluethree}{RGB}{165, 193, 224}
\definecolor{bluefour}{RGB}{141, 169, 200}

\usepackage{amsthm}
\theoremstyle{plain}
\newtheorem{proposition}{Proposition}

\newcommand\norm[1]{\lVert#1\rVert}
\newtheorem{remark}{Remark}
\usepackage{booktabs}
\usepackage{colortbl}
\definecolor{lightgray}{gray}{0.9}
\allowdisplaybreaks

\begin{document}

\title{Recovery Control in Replicated Systems\\through Autonomous Multiagent Rollout}

\author{Kim~Hammar and
        Yuchao~Li
\thanks{This research is supported by the Swedish Research Council; 2024-06436.}
\thanks{K. Hammar is with the Department of Computing, Imperial College, London, United Kingdom, \texttt{k.hammar@imperial.ac.uk}.}
\thanks{Y. Li is with the Fulton School of Engineering, Arizona State University, Tempe, AZ. \texttt{yuchaoli@asu.edu}.}
\thanks{Manuscript received Month Day, 2026.}}

{

\maketitle

\begin{abstract}
We study recovery control in replicated computing systems. Such systems consist of replicas that collectively provide a service to a client population. This redundancy enables the system to withstand failures provided that failed replicas are recovered faster than new failures occur. We show that the problem of deciding when to initiate recovery of selected replicas can be formulated as a partially observable Markov decision problem (POMDP) with a multiagent structure. We exploit this structure to apply a multiagent rollout method for approximating optimal control policies. Our method uses precomputed signaling information that reduces the need for replica coordination and facilitates parallel computations. Experiments show that our method scales to systems with up to 70 replicas and reduces costs compared to the recovery policies currently used in practice.
\end{abstract}

\begin{IEEEkeywords}
Reinforcement learning, multiagent, recovery.
\end{IEEEkeywords}

\IEEEpeerreviewmaketitle

\section{Introduction}
\IEEEPARstart{A}{s} our reliance on on-line services grows, there is an increasing demand for reliable systems that provide correct service without disruption. Research on fault-tolerant systems has almost a century-long history, with the seminal work by von Neumann \cite{vN56}, and Moore and Shannon \cite{MOORE1956191} in 1956. The early work focused on tolerance against hardware failures. Since then, the field has broadened to include tolerance against software bugs, operator mistakes, and cyberattacks; see e.g., Lamport et al. \cite{byzantine_generals_lamport} and Avizienis \cite{avivzienis1967design}. The common approach to building a fault-tolerant service is \textit{redundancy}, whereby the service is provided by a set of \textit{service replicas}. Through such redundancy, failed replicas can be substituted by healthy replicas as long as they can coordinate their service responses. This coordination problem is known as \textit{consensus}.\footnote{Here we refer to the classical notion of consensus from distributed systems theory, where a finite set of nodes must eventually agree on a value despite failures; see e.g., Lamport et al. \cite{byzantine_generals_lamport}. This is distinct from the notion of consensus studied in control theory, which focuses on states gradually converging to the same value over time; see e.g., Olfati-Saber et al. \cite{4118472}.}

Given a suitable consensus protocol, a distributed system with $N$ service replicas can tolerate up to $f < N$ faulty replicas, where a faulty replica can behave arbitrarily, i.e., Byzantine. For example, a faulty replica can stop responding to service requests (e.g., due to a power outage) or send incorrect responses to clients (e.g., due to a cyberattack). Tolerating failures in this setting means that despite the presence of up to $f$ faulty replicas, the system as a whole is still guaranteed to provide correct responses to service requests by clients.

While the system can only tolerate $f$ replicas being faulty \textit{simultaneously}, it can withstand \textit{any} number of failures provided that replicas are recovered quickly enough. However, deciding when to recover replicas is challenging because it requires balancing recovery costs against service requirements. Moreover, the failure state of a replica is often not fully observable due to limited monitoring capabilities. In addition, failures are not independent events, i.e., observing the failure of one replica increases the likelihood of additional failures. For instance, two replicas may be located in the same geographic region, so a regional outage could affect both simultaneously, or they might run the same operating system and therefore be vulnerable to the same types of cyberattacks.

\begin{figure}
  \centering
\scalebox{0.7}{
    \includegraphics{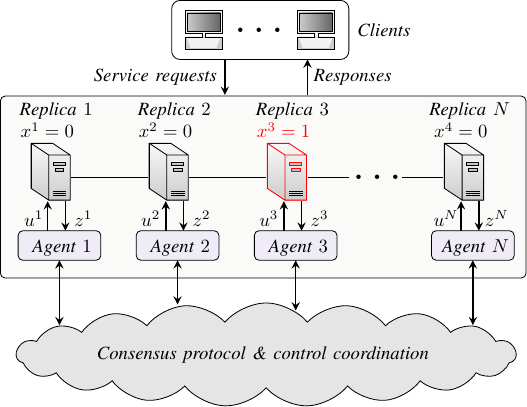}    
  }
\caption{Multiagent recovery control in a system with $N$ replicas that provide a service to clients, e.g., a web service. Each replica $i \in \{1,\hdots,N\}$ has a state $x^i$, which is $1$ if the replica is faulty (e.g., due to a cyberattack or operator mistake); $0$ otherwise. Replicas are coordinated through a consensus protocol that ensures correct service as long as at most $f < N$ replicas are faulty. Each replica $i$ is equipped with an agent that monitors the replica via the observation $z^i$ and applies the recovery control $u^i$.}\label{fig:example}
\end{figure}

Due to these challenges, current approaches to recovery control mostly rely on inefficient recovery policies, such as heuristic rule-based recovery (see e.g., \cite{4429182}) and periodic recovery (see e.g., \cite{pbft,zyzzyva}). While simple to implement, these policies do not adapt well to varying failure patterns and can waste resources or delay recovery. For example, in a periodic recovery scheme, a replica might be recovered every week even if it is functioning correctly, leading to unnecessary recovery costs and/or slow responses to failure events.

In this paper, we address these limitations and present a scalable method for approximating optimal recovery policies. Compared to heuristic approaches, our method can significantly reduce operational costs by avoiding unnecessary recoveries and minimizing service disruption. In particular, we formulate the recovery problem as a partially observable Markov decision problem (POMDP) with a multiagent structure, i.e., a POMDP where the control consists of multiple components. We exploit this structure to apply a multiagent rollout method where each replica is controlled by an \textit{agent} that decides when to initiate recovery based on partial observations; see Fig. \ref{fig:example}. This method is applicable to a broad class of POMDPs and extends beyond the recovery problem studied in this paper.

In the simplest instantiation of our method, each agent decides when to initiate recovery of its replica by performing a lookahead optimization using a base policy (e.g., a periodic recovery policy) to simulate future controls and predict the controls of the agents that have not yet selected their controls. The agents perform these optimizations sequentially (one at a time), where each agent is informed of the controls selected by the preceding agents. This sequential decomposition reduces the complexity of the control minimization from $O(2^N)$ to $O(N)$, where $N$ is the number of replicas. We also present an extension of our method wherein each agent selects its recovery control autonomously based on precomputed \textit{signaling} information. This extension allows the agents to compute their controls in parallel, which enables our method to scale to systems where standard rollout is computationally intractable.

The multiagent rollout formulation that forms the basis of our method was first presented by Bertsekas; see textbooks \cite{bertsekas2021rollout, bertsekas2024reinforcement} and paper \cite{9317713}. It is based on an approach for problem reformulation to trade off control space complexity with state space complexity, first proposed by Bertsekas and Tsitsiklis \cite[Section 6.1.4]{BertsekasTsitsiklis96}. Multiagent rollout methods have since been applied across a wide range of domains, including coordination of warehouse robots \cite{EMANUELSSON20233027}, repair scheduling for fleets of robots \cite{pmlr-v155-bhattacharya21a}, taxi routing \cite{10611063}, resource allocation \cite{10575707}, target tracking \cite{9658603}, antenna tilt optimization \cite{10624777}, multivehicle routing \cite{10.5555/3635637.3663054}, and robotic arm path planning \cite{li2026line}.

These studies highlight the versatility of multiagent rollout for problems where distributed decisions must be coordinated in real time. To our knowledge, we are the first to apply multiagent rollout to the recovery control problem studied in this paper. Moreover, we characterize theoretical conditions under which our method has an approximate cost-improvement property. Two further distinctions are that our method uses signaling and that we evaluate it both in simulation and on a physical testbed. By contrast, the referenced studies are limited to simulation and most of them do not use signaling.

We note that multiagent rollout relates to the theory of teams (see e.g., Marschak \cite{b0f02537-1883-352a-8f1c-e65595229fa6}), decentralized control (see e.g., Sandell et al. \cite{1101704}), and the notion of person-by-person optimality (see e.g., Bauso et al. \cite{4586577}), all of which are well-developed research areas with a long history. For recent developments in these areas, see e.g., Ballotta et al. \cite{10018269}, Doerr et al. \cite{9354997}, Trajanovski et al. \cite{7096986}, and Gattami \cite{7390057}. While related to this paper on a methodological level, all of these references study fundamentally different problems than the distributed recovery control problem studied in this paper.

In the context of distributed recovery control in replicated computing systems, the only prior work that adopts a decision-theoretic formulation rather than a heuristic one is the approach by Hammar and Stadler \cite{dsn24_hammar_stadler,10.1007/978-3-031-74835-6_1}. However, their method assumes that failures are independent across replicas, an assumption that does not hold in our setting. Other models that have similarities with the one studied in this paper include the singleagent recovery-control model by Kreidl and Frazier \cite{1282171} and the distributed failure-detection model by Kreidl, Tsitsiklis, and Zoumpoulis \cite{5667065}. Compared to those papers, the main difference is that our model is designed for distributed recovery control, whereas the referenced papers study centralized control and decentralized detection, respectively. Finally, we note that a large body of research has studied fault detection and recovery of control systems, see e.g., the work by Teixeira et al. \cite{7993039}, Burbano et al. \cite{10132877}, Xu et al. \cite{9329102}, and Gertler \cite{9163}. Our work differs from these papers by studying recovery control of a computing infrastructure (e.g., a cloud system) rather than a physical process (e.g., a power plant).

In summary, our main contributions are as follows:
\begin{itemize}
\item We develop a novel method for distributed recovery control in replicated computing systems. Our method is based on autonomous multiagent rollout with signaling.
\item We characterize the computational complexity of our method and define theoretical conditions under which it has an approximate cost improvement property.
\item We validate our method through extensive simulations and testbed experiments. The results show that our method improves scalability compared to other rollout methods and improves performance compared to the recovery control policies currently used in practice.
\end{itemize}  

\section{System Model}\label{sec:system_model}
We consider a replicated system with $N$ \textit{nodes}, which are connected through an \textit{authenticated network}, i.e., a network where nodes can verify the sender of a message using digital signatures. We assume an upper bound on the communication delay between nodes, but allow for temporary periods of instability where the bound is violated. In other words, we consider the \textit{partially synchronous system model}.\footnote{For technical details about the partially synchronous system model, we refer to the textbook by Cachin et al. \cite{cachin} and the paper by Dwork et al. \cite{partial_synchrony}.} Within this model, any message sent between two healthy nodes is eventually delivered. However, the time to deliver a message cannot be predicted. As is well known, this is the weakest model in which distributed consensus is theoretically feasible; see Fischer, Lynch, and Paterson for a proof \cite[Thm. 1]{flp}.

Each node is segmented into two domains: a \textit{privileged domain}, which can only fail by crashing, and an \textit{application domain}, which can fail in arbitrary (i.e., Byzantine) ways.\footnote{The separation between the privileged and application domains can be realized in several ways. One option is to use a secure coprocessor to execute the privileged domain; see e.g., \cite{pbft}. Another option is implementing the privileged domain using dedicated hardware modules, such as a smartcard or an FPGA; see e.g., \cite{bft_systems_survey}. A third option, which does not require special hardware, is to separate the two domains using a secure virtualization layer that can be formally verified; see e.g., \cite{virt_replica}.} In other words, we consider the \textit{hybrid failure model}, as originally proposed by Correia et al. \cite{1180168}. The privileged domain runs a control \textit{agent} and the application domain runs a \textit{service replica}. The agent monitors the replica through failure alerts (e.g., from an anomaly detector) and decides when to recover the replica by restarting it with a new operating system; see Fig.~\ref{fig:node}. The replicas are coordinated through a consensus protocol, which guarantees correct service to clients if no more than $f < N$ replicas are faulty at the same time.

While the control agents themselves can crash, they cannot exhibit arbitrary (Byzantine) behavior like the service replicas. As a consequence, crashed agents can be detected and replaced using a timeout-based crash detector. Recovering service replicas is more difficult because their Byzantine behavior can mask or mimic correct operation, making it hard to distinguish faulty replicas from healthy ones. For this reason, our focus is on the more challenging problem of recovering the service replicas. We formalize this problem in the next section.

\begin{figure}[H]
  \centering
\scalebox{0.6}{
    \includegraphics{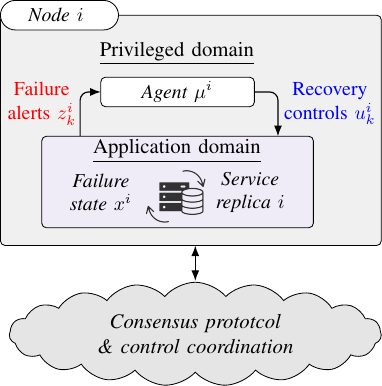}
  }
\caption{Architecture of a node in our system model. The node is segmented into a privileged domain that runs a control agent and an application domain that runs a service replica. This segmentation allows us to consider different failure assumptions: the privileged domain can only fail by crashing, while the application domain can fail in arbitrary (i.e., Byzantine) ways.}\label{fig:node}
\end{figure}
\section{Formulating Recovery Control as a POMDP}\label{sec:recovery_pomdp}
We formulate recovery control in a replicated computing system as an infinite-horizon POMDP with a discounted cost criterion. We denote the set of states by $X$, the set of controls by $U$, and the set of observations by $Z$, all being finite. Transitions from state $x$ to state $y$ under control $u$ occur at discrete times $k$ according to transition probabilities $p_{xy}(u)$. Each transition is associated with a cost $g(x,u,y)$ and an observation $z$, which is generated with probability $p(z \mid y, u)$.

In the context of our system model, the state at time $k$ is a vector $x_k=(x^1_k,\hdots,x^N_k)$, where $x^i_k=1$ if replica $i$ is faulty and $x^i_k=0$ otherwise. The initial state is $x_0=(0,\hdots,0)$. Similarly, the control is a vector $u_k=(u^1_k,\hdots,u^N_k)$, where $u^i_k=1$ means to recover replica $i$ and $u^i_k=0$ means to wait. If replica $i$ is faulty ($x_k^i=1$), then it remains so until recovery ($u_k^i=1$), at which point the state $x_{k+1}^i$ is set to $0$. Otherwise, we define the probability that the $i$th replica fails as
\begin{align}
\min\left\{p_{\mathrm{F}}\left(1+ \sum_{j=1}^NA_{ji}x^{j}_k\right), 1\right\},\label{eq:failure_prob}
\end{align}
where $p_{\mathrm{F}} \in (0,1)$ is a configurable parameter that models the propensity for failures and $A$ is an $N\times N$ binary matrix that encodes failure dependencies between replicas. Specifically, $A_{ji}=1$ means that the failure of replica $j$ increases the probability that replica $i$ will fail. Conversely, $A_{ji}=0$ means that the failure probability of replica $i$ is independent of the state of replica $j$. In our model, we set the diagonal entries of $A$ to one, i.e., $A_{ii}=1$. Conceptually, $A$ can be viewed as the adjacency matrix of a graph where an edge from node $j$ to node $i$ means that $i$ is failure-dependent on $j$; see Fig.~\ref{fig:failure_prob}.

\begin{figure}
  \centering
\scalebox{0.7}{
    \includegraphics{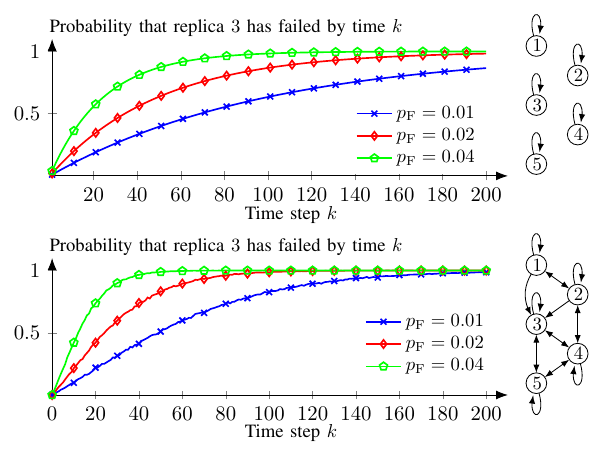}
  }
\caption{Probability that replica $3$ has failed by time $k$ in a system with $N=5$ replicas, assuming no recoveries occur. The two plots relate to different configurations of the failure-dependency matrix $A$, as illustrated by the graphs to the right of each plot. An edge from node $j$ to node $i$ in the graph means that $A_{ji}=1$. Curves relate to different values of the parameter $p_{\mathrm{F}}$; cf.~Eq.~\eqref{eq:failure_prob}.}\label{fig:failure_prob}
\end{figure}

Due to limited monitoring capabilities, the state $x_k$ is not observable. Rather, the control agents observe the vector $z_k = (z^1_k, \hdots, z^N_k)$, where $z^i_k \in \{0,1,\hdots,w\}$ models the number of failure alerts related to replica $i$ at time $k$. The distribution of these alerts is denoted by $p(z \mid y, u)$ and depends on the system state and the applied control. We estimate this distribution based on empirical measurements of failure alerts on our testbed; see Section~\ref{sec:experiment_setup} for details.

When selecting controls, the objective is to prevent unmitigated failures, minimize service disruption, and avoid unnecessary recovery controls. To capture these tradeoffs, we model the stage cost of the POMDP as
\begin{align}\label{eq:stage_cost}
g(x, u, y) = \overbrace{\lambda\phi(x,u)}^{\text{disruption cost}} + \sum_{i=1}^{N}\Big(\overbrace{\eta x^i(1-u^i)}^{\text{failure cost}} + \overbrace{u^i(1-x^i)}^{\text{recovery cost}}\Big),
\end{align}
where $0 \leq \eta,\lambda < \infty$ are weighting factors and $\phi$ is given by
\begin{align}\label{eq:downtime_cost}
  \phi(x,u) =
  \begin{dcases}
    1 & \text{ if }\left(\sum_{i=1}^{N}(x^i + (1-x^i)u^i) \right) > f,\\
    0 & \text{otherwise.}\\    
  \end{dcases}  
\end{align}
This stage cost expresses that costs are incurred for unmitigated failures ($x^i=1$ and $u^i=0$), unnecessary recovery controls ($u^i=1$ and $x^i=0$), and service disruption, which occurs when the number of replicas that are either faulty or being recovered exceeds the tolerance threshold $f$; cf.~Eq.~\eqref{eq:downtime_cost}.

\subsection*{Belief Space Formulation}
While the preceding formulation involves imperfect state information, it can be formulated as an equivalent problem with perfect state information; see e.g., \cite{ASTROM1965174}. In this formulation, the system is described by the belief state $b_k$, which is a vector with $2^N$ components,\footnote{The belief is defined over the state space $X=\{0,1\}^N$ to capture dependencies among the failure states of different replicas. When these states are independent, a compact belief representation of dimension $N$ is possible.} where each component $b_k(x)$ is the conditional probability that the state is $x$ at time $k$, given the history of controls and observations. We use $b^i_k$ to denote the belief that the state of replica $i$ is $1$, i.e.,
\begin{align}
b_k^i = \sum_{x\in X^i}b_k(x),\label{eq:marginal_b}
\end{align}
where $X^i\subset X$ is the subset of states where $x^i=1$.

The belief state belongs to the belief space $B$ and is recursively updated through a belief estimator $F$ as
\begin{align}
b_{k}&= F(b_{k-1}, u_{k-1}, z_k). \label{eq:belief_estimator}
\end{align}
The principled way to define $F$ is via the recursion
\begin{align}
b_{k}(y)&= \frac{p(z_{k} \mid y, u_{k-1})\sum_{x\in X}b_{k-1}(x)p_{xy}(u_{k-1})}{\sum_{x^{\prime} \in X}\sum_{y^{\prime} \in X}p(z_{k} \mid y^{\prime}, u_{k-1})b_{k-1}(x^{\prime})p_{x^{\prime}y^{\prime}}(u_{k-1})}.\label{eq:bayes_belief}
\end{align}

We adopt the belief-space formulation and consider recovery policies $\mu$ that map the belief space $B$ to the control space $U$. The cost function of such a policy is defined as
\begin{align}
J_{\mu}(b_0)=\lim_{H\rightarrow \infty}E\left\{\sum_{k=0}^{H-1}\alpha^k\hat{g}\big(b_k, \mu(b_k)\big)\right\},\label{eq:pomdp_minimization}
\end{align}
where $E\{\cdot\}$ denotes the expected value, $\alpha \in (0,1)$ is a discount factor, and the stage cost $\hat{g}(b,u)$ is defined as
\begin{align}
\hat{g}(b,u)=&\sum_{x\in X}b(x)\sum_{y \in X}p_{xy}(u)g(x,u,y)\label{eq:hat_g}.
\end{align}
The optimal cost function $J^{*}$, derived by optimizing over all possible policies $\mu$, uniquely satisfies the Bellman equation
\begin{align}
J^{*}(b) &= \min_{u \in U}\left[\hat{g}(b, u) + \alpha \sum_{z \in Z}\hat{p}(z \mid b, u)J^{*}(F(b,u,z))\right],\label{eq:optimal_cost}
\end{align}
where the probability $\hat{p}(z \mid b, u)$ is defined as
\begin{align}
\hat{p}(z \mid b,u)=&\sum_{x \in X}b(x)\sum_{y \in X}p_{xy}(u)p(z \mid y,u).\label{eq:hat_p}
\end{align}
We say that a policy $\mu^{*}$ is optimal if $J_{\mu^*}=J^{*}$. Although such a policy exists (see e.g., \cite[Thm. 7.6.1]{krishnamurthy_2016} or \cite[§ 5.6]{bertsekas2012dynamic}), there are no efficient algorithms to obtain it. Consequently, approximations are required in practice. We present two such approximation methods in the following sections.

\section{Centralized Control: Singleagent Rollout}
While the problem in the preceding section can in principle be solved using dynamic programming, the amount of required computation grows exponentially with the number of service replicas; see Fig. \ref{fig:spaces}. For this reason, we use approximation techniques. In particular, we use the \textit{rollout} methodology as originally formulated by Bertsekas \cite{bertsekas2021rollout}.

\begin{figure}[H]
  \centering
\scalebox{0.7}{
    \includegraphics{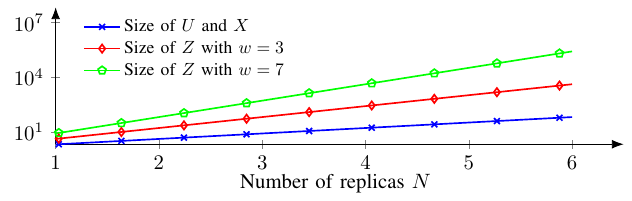}
  }
\caption{Size of the control, state, and observation spaces as a function of the number of replicas $N$. The value $w+1$ is the number of possible failure alerts that can be observed at each replica.}\label{fig:spaces}
\end{figure}

The starting point of our rollout method is a \textit{base policy} $\mu$, which can be chosen freely. It can, for example, be designed based on off-line optimization (see e.g., \cite{tifs_25_HLALB}) or heuristics (see e.g., \cite{10.1109/TRO.2023.3347128}). In the context of recovery control, a natural choice may be a policy that recovers replicas periodically. Such a recovery policy was used by Castro and Liskov \cite{pbft}.

Regardless of the type of base policy, we can start from a belief $b_k$ and simulate POMDP trajectories of the form 
\begin{align*}
b_k, u_k, b_{k+1}, u_{k+1}, \hdots, b_{k+m-1},u_{k+m-1}, b_{k+m},
\end{align*}
where $m \geq 1$ is the \textit{rollout horizon} and each control is given by $u_l=\mu(b_l)$. Such simulations allow us to estimate the expected cost-to-go of the base policy from any belief $b_k$ as
\begin{align}
\Tilde{J}_{\mu}(b_{k})\!\!&=\!\!\frac{1}{L}\sum_{s=1}^{L}\sum_{l=k}^{k+m-1}\alpha^{l-k}\hat{g}\big(b_{l,s},\mu(b_{l,s})\big) + \alpha^{m}\Tilde{J}(b_{k+m,s}),\label{eq:rollout_approximation}                                  
\end{align}
where $L \geq 1$ is the number of simulations. Here $(b_{k,s},\dots,b_{k+m,s})$ is the belief trajectory of the $s$th simulation and $\Tilde{J}$ is a (bounded) \textit{cost function approximation}, which can be designed based on off-line approximation techniques. For example, $\Tilde{J}$ can be obtained through aggregation methods (see e.g., \cite{tifs_25_HLALB,li2025feature,Hammar_Alpcan_2026}) or deep learning (see e.g., \cite{10575707}).

Given the base policy $\mu$ and its cost-to-go estimate $\Tilde{J}_{\mu}$ [cf.~Eq.~\eqref{eq:rollout_approximation}], we compute a \textit{rollout policy} $\tilde{\mu}$ through lookahead optimization as 
\begin{align}
  &\tilde{\mu}(b_k) \in \arg\min_{u_k \in U}\Bigg[\hat{g}(b_k,u_k) + \min_{\mu_{k+1},\hdots,\mu_{k+\ell-1}}\mathop{E}_{z_{k+1},\hdots,z_{k+\ell}}\nonumber\\
  &\quad\quad\quad\quad\bigg\{\sum_{j=k+1}^{k+\ell-1}\alpha^{j-k}\hat{g}(b_j, \mu_j(b_j)) + \alpha^{\ell}\Tilde{J}_{\mu}(b_{k+\ell})\bigg\}\Bigg] \label{eq:rollout},
\end{align}
where $\ell \geq 1$ is the lookahead horizon. This rollout policy possesses the fundamental \textit{cost improvement} property. To state this property formally, let us denote by $e:B \rightarrow \Re$ the unit function that takes the value $1$ identically on $B$, i.e., $e(b)\equiv1$. Consider the following optimization problem
\begin{subequations}
    \label{eq:epsilon_approx_rollout}
    \begin{align}
        \inf_{c\in \Re}& \quad \Vert \tilde J_\mu+ce-J_\mu\Vert \label{eq:infimum}\\
	\mathrm{s.\,t.} &\quad \min_{u \in U}\left[\hat{g}(b, u) + \alpha \sum_{z \in Z}\hat{p}(z \mid b, u)\tilde{J}_\mu(F(b,u,z))\right]\nonumber\\
    &\quad\leq \tilde{J}_\mu(b)+(1-\alpha)c,\quad \hbox{for all }b \in B,\label{eq:lyapunov_rollout}
    \end{align}
\end{subequations}
where $\norm{\cdot}$ is the maximum norm. Let $\epsilon$ denote the infimum of this optimization problem. Given this notation, we have the following result, which we prove in Appendix \ref{app:proof_prop_1}.
\begin{proposition}\label{prop:improvement} The rollout policy $\tilde{\mu}$ approximately improves the base policy $\mu$ with error at most $\alpha^\ell \epsilon$, i.e., 
$$J_{\tilde{\mu}}(b)\leq J_{\mu}(b)+\alpha^{\ell} \epsilon,\quad \hbox{for all }b \in B,$$
where $0\leq \epsilon<\infty$ is defined by the infimum \eqref{eq:epsilon_approx_rollout}.
\end{proposition}
Proposition \ref{prop:improvement} provides a performance guarantee for the rollout policy. It implies that the rollout policy will generally perform at least as well, and typically better than the base policy. Moreover, it can be shown that the performance gap between the rollout policy and the optimal policy decreases exponentially with the lookahead horizon $\ell$; see Bertsekas for a proof \cite[Prop. 5.1.1]{bertsekas2019reinforcement}. Another appealing property of the rollout method is that the computational cost can be scaled by tuning the number of lookahead steps ($\ell$), the number of simulations ($L$), and the rollout horizon ($m$), as stated below.
\begin{proposition}\label{prop:complexity_1}
If the belief estimator is implemented via Eq.~\eqref{eq:bayes_belief} and the evaluations of $p_{xy}(u)$, $p(z \mid y, u)$, $g(x,u,y)$, and $\tilde{J}(b)$ are constant-time operations, then the computation required by Eq.~\eqref{eq:rollout} is of order
\begin{align*}
O\left((2^N|Z|)^{\ell}Lm|X|^2\right),
\end{align*}
where $|\cdot|$ denotes the cardinality of a set.
\end{proposition}
\begin{proof}
Since the evaluations of $p_{xy}(u)$, $p(z \mid y, u)$, $g(x,u,y)$, and $\tilde{J}(b)$ are constant-time operations by assumption, it follows that the computations required by Eq.~\eqref{eq:bayes_belief} and Eq.~\eqref{eq:hat_g} are both of order $O(|X|^2)$. Furthermore, because Eqs.~\eqref{eq:rollout_approximation}--\eqref{eq:rollout} invoke Eq.~\eqref{eq:hat_g} at most once per invocation of Eq.~\eqref{eq:bayes_belief}, Eq.~\eqref{eq:hat_g} does not affect the order of the computation. As a result, the computation required by Eq.~\eqref{eq:rollout_approximation} is of order $O(Lm|X|^2)$. Now consider Eq.~\eqref{eq:rollout}. The lookahead can be represented as a tree with branching factor $|U||Z|=2^N|Z|$. For each edge of this tree, Eq.~\eqref{eq:rollout} invokes the belief estimator. The number of edges in the lookahead tree is $\sum_{j=1}^{\ell}(2^N|Z|)^j = O((2^N|Z|)^{\ell})$. Thus, the complexity of constructing the tree is of order $O((2^N|Z|)^{\ell}|X|^2)$. Moreover, for each of the beliefs reached after $\ell$ lookahead steps, Eq.~\eqref{eq:rollout} invokes Eq.~\eqref{eq:rollout_approximation}, which requires $O\left((2^N|Z|)^{\ell}Lm|X|^2\right)$ computations in total. Hence, the complexity of computing Eq.~\eqref{eq:rollout} is of order
\begin{align*}
\underbrace{O\left((2^N|Z|)^{\ell}Lm|X|^2\right)}_{\text{Cost-to-go computation}} + \underbrace{O\left((2^N|Z|)^{\ell}|X|^2\right)}_{\text{Lookahead tree construction}},
\end{align*}
which is dominated by the first term since $Lm \geq 1$.
\end{proof}
\begin{remark}\label{remark_1}
The belief estimator $F$ can be implemented using approximation schemes to reduce the computational cost. For example, if $F$ is implemented by a particle filter with $M$ particles, then the computation per belief update is of order $O(M)$ rather than $O(|X|^2)$, which reduces the complexity of computing Eq.~\eqref{eq:rollout} to $O\left((2^N|Z|)^{\ell}LmM\right)$.
\end{remark}
Despite the scalable computational complexity, practical computation of the rollout policy [cf.~Eq.~\eqref{eq:rollout}] involves two major challenges. First, the size of the control space $U$ is $2^N$, which means that the complexity of the minimization \eqref{eq:rollout} increases exponentially with the number of replicas $N$. Second, the minimization \eqref{eq:rollout} requires complete coordination among the service replicas (e.g., the selection of control $u^1$ depends on $u^2,\hdots,u^N$), which may be impractical. In the following section, we present a variant of rollout that reduces the computational complexity and the need for coordination. We refer to this rollout method as \textit{multiagent rollout}.

\begin{figure*}[!t]
\begin{subequations}\label{eq:multiagent_rollout}
\begin{align}
&\tilde{\mu}^1(b_k)\in \arg\min_{u^1_k \in \{0,1\}}\Bigg[\hat{g}\Big(b_k,\big(u^1_k,\mu^{2}(b_k),\mu^{3}(b_k),\dots,\mu^{N}(b_k)\big)\Big) + \alpha\mathop{E}_{z_{k+1}}\bigg\{\tilde{J}_{\mu}(b_{k+1})\bigg\}\Bigg],\\
&\tilde{\mu}^2(b_k)\in \arg\min_{u^2_k \in \{0,1\}}\Bigg[\hat{g}\Big(b_k,\big(\tilde{\mu}^1(b_k),u^2_k,\mu^{3}(b_k),\mu^{4}(b_k),\dots,\mu^{N}(b_k)\big)\Big) + \alpha\mathop{E}_{z_{k+1}}\bigg\{\tilde{J}_{\mu}(b_{k+1})\bigg\}\Bigg],\\
  \vdots\nonumber\\
&\tilde{\mu}^N(b_k)\in \arg\min_{u^N_k \in \{0,1\}}\Bigg[\hat{g}\Big(b_k,\big(\tilde{\mu}^1(b_k),\tilde{\mu}^2(b_k),\dots,\tilde{\mu}^{N-1}(b_k),u^N_k\big)\Big) + \alpha\mathop{E}_{z_{k+1}}\bigg\{\tilde{J}_{\mu}(b_{k+1})\bigg\}\Bigg].    
\end{align}
\end{subequations}
\end{figure*}
\begin{figure*}[!t]
\begin{equation}\label{eq:auto_multiagent_rollout}
\begin{aligned}
&\tilde{\mu}^i(b_k)\in \arg\min_{u^i_k \in \{0,1\}}\Bigg[\hat{g}\Big(b_k,\big(\widehat{\mu}^1(b_k),\dots,\widehat{\mu}^{i-1}(b_k),u^i_k, \mu^{i+1}(b_k),\dots,\mu^{N}(b_k)\big)\Big) + \alpha\mathop{E}_{z_{k+1}}\bigg\{\tilde{J}_{\mu}(b_{k+1})\bigg\}\Bigg].
\end{aligned}
\end{equation}
\end{figure*}
\section{Distributed Control: Multiagent Rollout}\label{sec:multiagent_rollout}
The key idea of multiagent rollout is to reformulate the control problem as an equivalent problem where the complexity of the control space $U$ is reduced at the expense of additional (observable) states, as originally proposed by Bertsekas and Tsitsiklis \cite[Section 6.1.4]{BertsekasTsitsiklis96}. This tradeoff is favorable since the complexity of rollout [cf.~Eq.~\eqref{eq:rollout}] is proportional to $|U|=2^N$ but is independent of the number of observable states; cf. Prop. \ref{prop:complexity_1} and Remark \ref{remark_1}. Specifically, we reformulate the original problem by introducing the intermediate states
\begin{align*}
(b_k,u^1_{k}), (b_k,u^1_k,u_k^2), \hdots, (b_k,u^1_{k},\hdots,u^{N-1}_{k}) && \text{for all }k.
\end{align*}
Given these intermediate states, the control components $u^1_k, \hdots,u^N_k$ are selected sequentially. That is, given the belief $b_k$, we start by selecting the control $u^1_k$. Then, given the belief $b_k$ and the control $u^1_k$, we select the control $u^2_k$, and so on; see Fig. \ref{fig:multiagent}. Conceptually, each control component $u_k^i$ is selected by a separate ``\textit{agent}''. It is evident that this reformulated problem is equivalent to the original, since any control that is possible in one problem is also possible in the other problem, while the cost structure of the two problems is the same.

\begin{figure}[H]
  \centering
  \scalebox{0.85}{
    \includegraphics{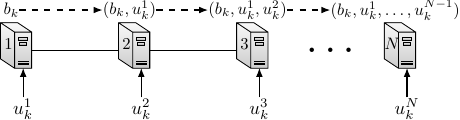}
  }
  \caption{Illustration of the reformulated (multiagent) problem. The control components $u^1_{k},\hdots,u^N_{k}$ are selected sequentially by separate agents, where the selection of $u^i_{k}$ is done in the intermediate state $(b_k, u^1_{k}, u^2_{k},\hdots,u^{i-1}_{k})$.}
  \label{fig:multiagent}
\end{figure}
\begin{remark}
The order in which the controls are selected is not restricted to $1,2,\hdots,N$ and can be optimized; see \cite{bertsekas2021rollout}.
\end{remark}

The motivation for the reformulated problem is that the computation required by the control minimization \eqref{eq:rollout} is reduced from being in the order of $O(2^N)$ to $O(N)$, where $N$ is the number of replicas. (This reduction is formalized in Prop.~\ref{prop:complexity_2} below.) Applying the rollout method [cf.~Eq.~\eqref{eq:rollout}] to the reformulated problem with $\ell=1$, we obtain Eq.~\eqref{eq:multiagent_rollout}, which is stated at the top of the next page. (We use $\mu^i(b_k)$ to denote the $i$th control component selected by the policy $\mu$.)

Despite the significant reduction in computational cost, the cost improvement property of multiagent rollout is maintained. To formalize this property, let $\tilde \mu$ denote the multiagent rollout policy and consider the following optimization problem
\begin{subequations}
    \label{eq:epsilon_approx_multi_rollout}
    \begin{align}
        \inf_{c\in \Re}& \quad \Vert \tilde J_\mu+ce-J_\mu\Vert\\
	\mathrm{s.\,t.} &\quad\left[\hat{g}\big(b, \tilde{\mu}(b)\big) + \alpha \sum_{z \in Z}\hat{p}\big(z \mid b, \tilde{\mu}(b)\big)\tilde{J}_\mu\Big(F\big(b,\tilde{\mu}(b),z\big)\Big)\right]\nonumber\\
    &\quad\leq \tilde{J}_\mu(b)+(1-\alpha)c,\quad \hbox{for all }b.\label{eq:lyapunov_multi_rollout}
    \end{align}
\end{subequations}
Denote by $\hat \epsilon$ the infimum above. Given this notation, we have the following result, which we prove in Appendix \ref{app:proof_prop_3}.  
\begin{proposition}\label{prop:multiagent_rollout_improvement}
The multiagent rollout policy $\tilde{\mu}$ computed via Eq.~\eqref{eq:multiagent_rollout} approximately improves the base policy $\mu$ with error at most $\hat \epsilon$, i.e., 
$$J_{\tilde{\mu}}(b)\leq J_{\mu}(b)+ \hat\epsilon,\quad \hbox{for all }b \in B,$$
where $0\leq \hat\epsilon<\infty$ is defined by the infimum \eqref{eq:epsilon_approx_multi_rollout}.
\end{proposition}

Compared with the performance improvement of standard rollout with one-step lookahead, the error term $\alpha\epsilon$ in Prop.~\ref{prop:improvement} is replaced by $\hat \epsilon$. Since the feasible set defined by Eq.~\eqref{eq:lyapunov_rollout} is larger than that defined by Eq.~\eqref{eq:lyapunov_multi_rollout}, we have $\epsilon \leq \hat \epsilon$. Therefore, the improvement guarantee of multiagent rollout, as stated in Prop.~\ref{prop:multiagent_rollout_improvement}, is weaker than that of standard rollout with one-step lookahead given in Prop.~\ref{prop:improvement}. However, our experimental evaluation in \S\ref{sec:exp_eval} shows that both standard rollout and multiagent rollout attain comparable performance in practice.

In addition to preserving the approximate cost improvement property, multiagent rollout significantly reduces the computational cost. In particular, the factor $|U|=2^N$ in Prop.~\ref{prop:complexity_1} is reduced to $2N$, as stated below. The proof follows the same arguments as the proof of Prop.~\ref{prop:complexity_1} and is omitted for brevity.

\begin{proposition}\label{prop:complexity_2}
If the belief estimator is implemented via Eq.~\eqref{eq:bayes_belief} and the evaluations of $p_{xy}(u)$, $p(z \mid y, u)$, $g(x,u,y)$, and $\tilde{J}(b)$ are constant-time operations, then the computation required by Eq.~\eqref{eq:multiagent_rollout} is of order
\begin{align*}
O\left(N|Z|Lm|X|^2\right).
\end{align*}
\end{proposition}

While multiagent rollout reduces the amount of required computation, it requires significant coordination among the agents, which means that the computations cannot be parallelized. Specifically, multiagent rollout requires that the optimization of the control component $u^i$ takes into account the controls $\tilde\mu^1(b_k),\tilde\mu^2(b_k),\hdots,\tilde\mu^{i-1}(b_k)$; cf. Eq.~\eqref{eq:multiagent_rollout}. In practice, it may be more practical to have agent $i$ select its control before knowing the rollout controls of some of the agents $1,\hdots,i-1$ and instead use the controls $\widehat{\mu}^1(b_k), \hdots, \widehat{\mu}^{i-1}(b_k)$ in their place, where $\widehat{\mu}$ is a \textit{signaling policy} that embodies agent coordination. In other words, the agents compute their controls in parallel and use the signaling policy to predict the controls of the preceding agents. We call this method \textit{autonomous multiagent rollout} and define it in Eq.~\eqref{eq:auto_multiagent_rollout} at the top of the page.

The benefit of this method is that it provides additional speedup through parallelization. However, it does not possess the cost improvement property. Hence, \textit{autonomous multiagent rollout} may work well for some problems and not work well for other problems. The major determinants of the performance of this method are a) to what extent agent coordination is needed to select effective controls; and b) the accuracy of the signaling policy $\widehat{\mu}$. In practice, the signaling policy can be designed off-line based on heuristics or approximation architectures, such as neural networks. We evaluate different designs of the signaling policy in the next section.
 
\section{Experimental Evaluation}\label{sec:exp_eval}
In this section, we present an experimental evaluation of the rollout methods described above. Our study focuses on the tradeoff between computational time and cost.
\begin{figure*}
  \centering
\scalebox{0.62}{
    \includegraphics{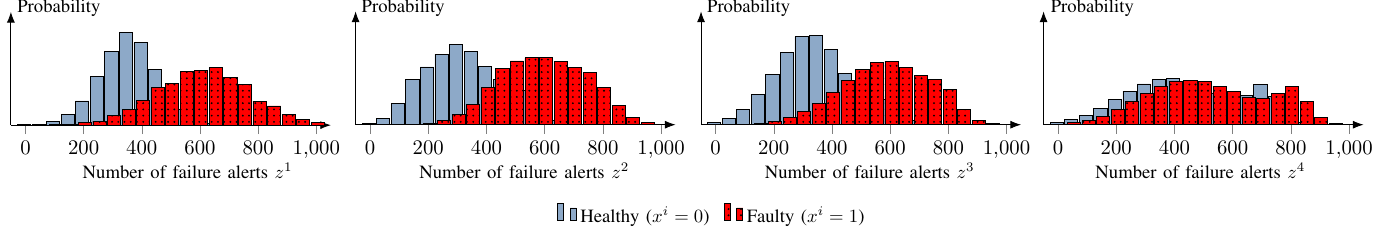}
  }
\caption{Empirical observation distributions based on our testbed measurements. The measurements are available at \texttt{https://github.com/Kim-Hammar/csle}. We plot the distributions as $p(z^i \mid x^i)$ rather than $p(z \mid x, u)$ since the distributions across replicas are independent given the local state $x^i$ and do not depend on the control $u$. The figure shows the observation distribution of four example replicas, i.e., $i \in \{1,2,3,4\}$. For details about our testbed setup, see Appendix \ref{app:testbed}.}\label{fig:obs_dist}
\end{figure*}

\subsection{Experiment Setup}\label{sec:experiment_setup}
All computations are performed on an M4 pro processor with $12$ CPU cores. When computing the expected cost of a policy, we take the average based on $100$ simulated POMDP executions of length $100$. We instantiate the POMDP described in Section~\ref{sec:recovery_pomdp} as follows. We set the failure probability to be $p_{\mathrm{F}}=0.05$ and the dependency matrix $A$ to be the adjacency matrix of a random Erdős–Rényi graph with parameter $p=0.5$. Similarly, we define the tolerance threshold to be $f = \lfloor\frac{N-1}{2}\rfloor$; cf.~Eq.~\eqref{eq:downtime_cost}. Moreover, we define the weighting factors in the stage cost to be $\eta=0.2$ and $\lambda = 1.5$; cf.~Eq.~\eqref{eq:stage_cost}. Lastly, we define the discount factor to be $\alpha=0.95$.

We consider two variants of the belief estimator $F$ in Eq.~\eqref{eq:belief_estimator}: the filter in Eq.~\eqref{eq:bayes_belief} and a particle filter defined as 
\begin{align}
\hat{b}_k(x) = \frac{1}{M}\sum_{s=1}^M\delta_{x\hat{x}_k^{s}}, && \text{for all }x \in X,\label{eq:estimate_belief}
\end{align}
where $\delta_{xy}=1$ if $x=y$ and $\delta_{xy}=0$ if $x\neq y$. The states (particles) $\hat{x}_k^{1},\hdots,\hat{x}_k^{M}$ are sampled with probability proportional to the numerator in Eq.~\eqref{eq:bayes_belief}. In our experiments, we use $M=50$ particles. Moreover, when using particle filtering, we approximate the expected values in Eq.~\eqref{eq:rollout} and Eq.~\eqref{eq:multiagent_rollout} by taking the average values of $100$ simulations.

We instantiate the rollout method [cf.~Eq.~\eqref{eq:rollout}] with $L=10$ simulations and define the base policy $\mu$ as
\begin{align}\label{eq:base_policy}
  \mu^i(b_k) =
  \begin{dcases}
    1 & b^i_k > 0.9,\\
    0 & \text{otherwise,}
  \end{dcases}
&& i \in \{1,\hdots,N\}.  
\end{align}
Similarly, we define the cost function approximation in Eq.~\eqref{eq:rollout_approximation} as $\Tilde{J}(b) = \hat{g}(b, \mu(b))$. Finally, we use the rollout horizon $m=5$ and the lookahead horizon $\ell=1$.

\vspace{2mm}
\noindent\textit{\textbf{Signaling policies.}} We consider two signaling policies:
\begin{itemize}
  \item \textit{Base signaling policy $\widehat{\mu}(b)$.}
\begin{itemize}
\item This signaling policy uses the base policy $\mu$ to predict the controls of the other agents.
\end{itemize}
\item \textit{Neural network signaling policy $\widehat{\mu}_{\theta}(b)$.}
\begin{itemize}
\item This signaling policy is represented by a neural network that is trained off-line to imitate the multiagent rollout policy; cf.~Eq.~\eqref{eq:multiagent_rollout}. See Appendix \ref{app:neural_training} for details about the neural network training.
\end{itemize}
\end{itemize}

\vspace{2mm}
\noindent\textit{\textbf{System identification.}} We identify the observation distribution $p(z \mid y,u)$ of the POMDP from measurement data obtained by deploying $N=50$ service replicas on our testbed and running a sequence of emulated failure events; see Fig.~\ref{fig:dt_serverrack}. We define the length of a time step in our testbed to be $30$ seconds. During each step, we emulate failure events at a subset of replicas. We then measure the observation $z_k$ by reading log files. We repeat this procedure for $24,386$ time steps and use the empirical distribution of failure alerts to define the observation distribution $p(z \mid y, u)$ in the POMDP. The observation space is defined as $Z = \{0,1,\hdots,999\}^N$, i.e., the maximum number of alerts per replica is $999$. For details about our testbed setup, see Appendix \ref{app:testbed}.

Figure \ref{fig:obs_dist} shows the estimated observation distribution $p(z \mid y, u)$. We note that the observation distributions are conditionally independent across replicas. The reason for this conditional independence is that the failure alerts are generated by local monitoring systems that depend only on the local state of each replica. Moreover, we see in Fig.~\ref{fig:obs_dist} that the distributions differ between replicas. The reason for these differences is that while the replicas run the same service, they have different software versions. We also note that the observation distributions of a replica $i$ in the healthy state ($x^i=0$) and the faulty state ($x^i=1$) overlap. However, the observation distribution in the faulty state tends to put more probability mass on large values of $z^i$, as expected.

\subsection{Evaluation results}
In this section, we present our evaluation results. We start by presenting the results from two simulation studies. In the first study, we compare the scalability of autonomous multiagent rollout with the standard multiagent rollout and the singleagent version. In the second study, we compare the cost of the rollout policy with that of the base policy. After these two simulation studies, we present the results of a comprehensive testbed evaluation where we compare the performance of the rollout policies with that of the periodic policies used in practice. (Details about our testbed are provided in Appendix \ref{app:testbed}.)

\begin{figure}
  \centering
    \scalebox{0.045}{
      \includegraphics{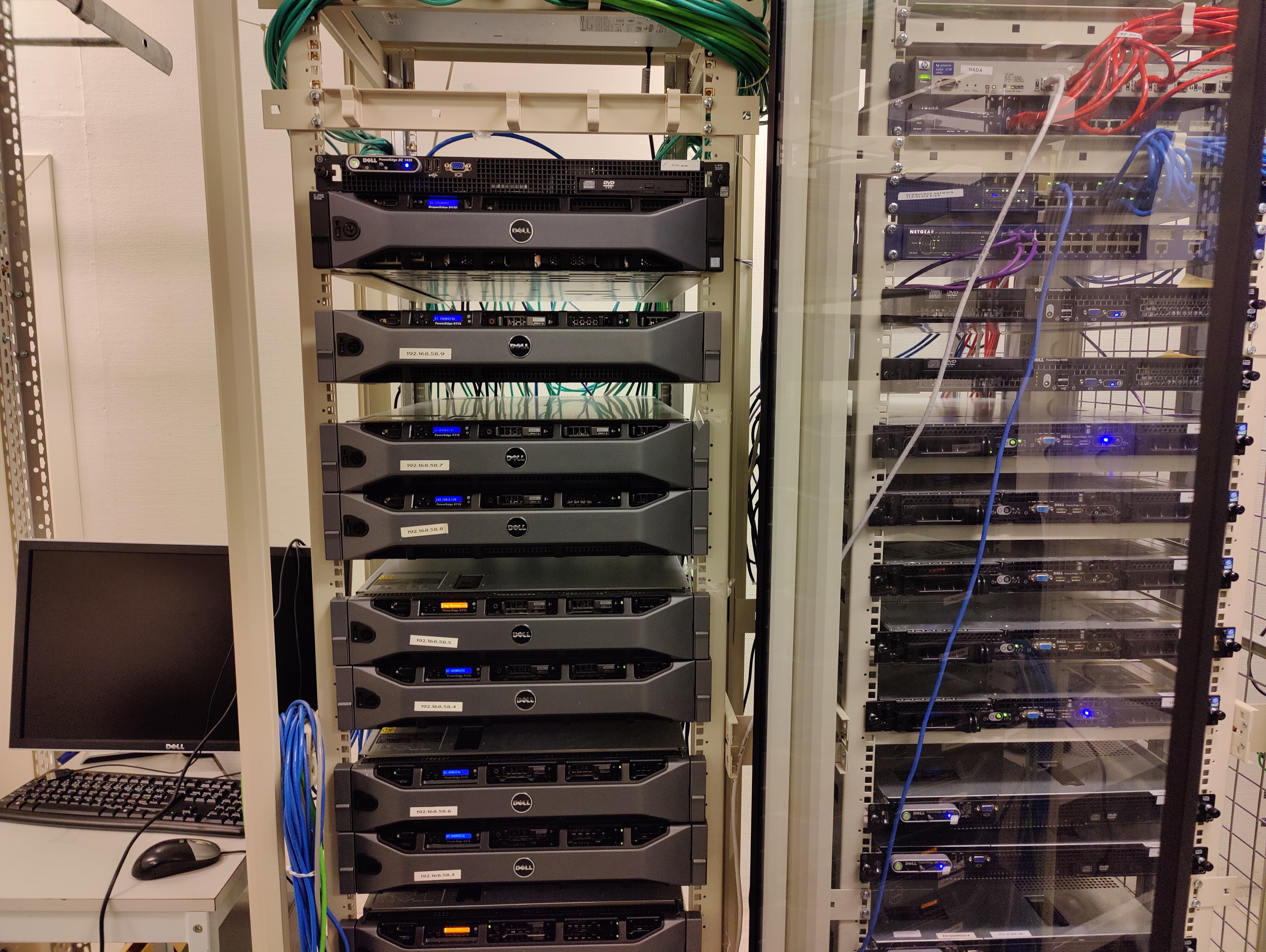}
    }
    \caption{The testbed that we use for system identification and evaluation \cite{dsn24_hammar_stadler}. It consists of physical servers connected through an IP network. We use these servers to run service replicas (in virtual containers), emulate failure events (in the form of cyberattacks, network partitions, and power outages), measure failure alerts, and execute recovery controls. The software platform that we use for emulation is available at \texttt{https://github.com/Kim-Hammar/csle}.}
    \label{fig:dt_serverrack}
\end{figure}

\vspace{2mm}
\noindent\textit{\textbf{Scalability evaluation.}} As described in Section \ref{sec:multiagent_rollout}, the main benefit of multiagent rollout compared to singleagent rollout is that it scales to systems with a large number of service replicas. Figure \ref{fig:scale} shows the compute time of different rollout methods as a function of the number of replicas $N$. 

We observe that singleagent rollout [cf.~Eq.~\eqref{eq:rollout}] with the belief estimator in Eq.~\eqref{eq:bayes_belief} becomes computationally intractable for systems with more than $N=3$ replicas. The reason for this intractability is the combinatorial explosion of the observation space as $N$ grows. In particular, as expressed in Prop. \ref{prop:complexity_1}, the branching factor of a lookahead step of the rollout method is $|U||Z|$, where $|Z|=1000^N$ in our case.

The particle filter in Eq.~\eqref{eq:estimate_belief} in combination with simulation-based estimation of the expected value in Eq.~\eqref{eq:rollout} allows to circumvent this computational intractability, as shown in the red curve of the upper plot of Fig.~\ref {fig:scale}. However, even with these approximation techniques, singleagent rollout quickly becomes impractical as $N$ grows, with compute times increasing steeply already for $N>8$. The reason for this rapid increase is the combinatorial explosion of the control space, which has size $|U|=2^N$, as illustrated in Fig.~\ref{fig:spaces}. 

In contrast, the compute times of multiagent rollout [cf.~Eq.~\eqref{eq:multiagent_rollout}] grow moderately with the number of replicas (near linear), as shown in the middle plot of Fig. \ref{fig:scale}. Autonomous multiagent rollout [cf.~Eq.~\eqref{eq:auto_multiagent_rollout}] further reduces the compute times through parallelization. As shown in the lower plot of Fig. \ref{fig:scale}, this parallelization results in a compute time that is nearly constant as the number of replicas ($N$) increases. 

\begin{figure}
  \centering
\scalebox{0.65}{
    \includegraphics{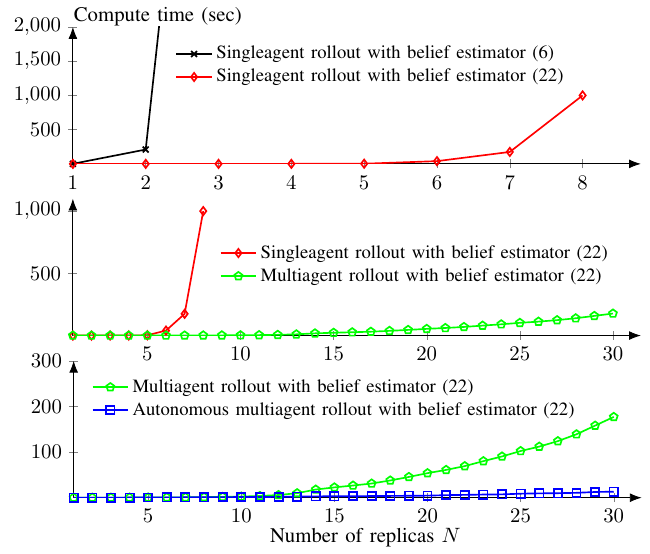}    
  }
  \caption{Compute time of different rollout methods as a function of the number of replicas $N$. All rollout methods are instantiated with lookahead horizon $\ell=1$, rollout horizon $m=5$, and $L=10$ simulations; cf.~Eq.~\eqref{eq:rollout}. The particle filter is instantiated with $M=50$ particles; cf.~Eq.~\eqref{eq:estimate_belief}. The numbers correspond to the average compute time across $3$ executions.}\label{fig:scale}
\end{figure}

\vspace{2mm}
\noindent\textit{\textbf{Cost evaluation.}} 
To compare the cost incurred by the rollout policies and the base policy, we consider two scenarios:
\begin{itemize}
\item \textit{Simulation scenario 1}: In this scenario, we configure the POMDP as described in Section~\ref{sec:experiment_setup}.
\item \textit{Simulation scenario 2}: In this scenario, we configure the POMDP as described in Section~\ref{sec:experiment_setup} except that we set the parameter $\lambda$ to $20$ instead of $1.5$. This change means that the stage cost [cf.~Eq.~\eqref{eq:stage_cost}] is dominated by a term that depends on the state and control of all $N$ replicas. As a consequence, the importance of coordination among the agents is larger in this scenario.
\end{itemize}    
\begin{figure}
  \centering
\scalebox{0.65}{
    \includegraphics{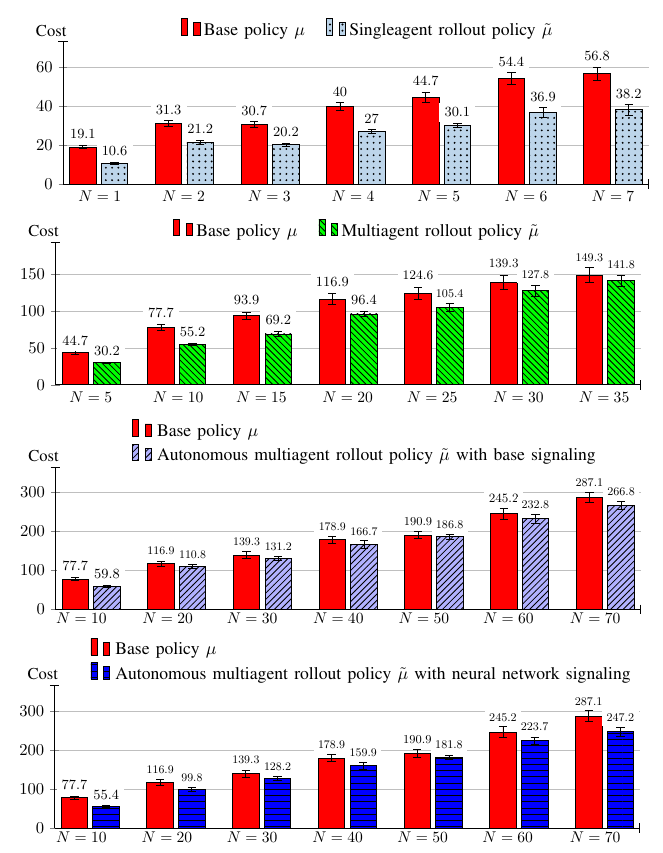}    
  }
  \caption{Simulation scenario 1: Comparison between the cost of the base policy $\mu$ [cf.~Eq.~\eqref{eq:base_policy}] and different rollout policies $\tilde{\mu}$. For all policies, we use particle filtering to implement the belief estimator; cf.~Eq.~(\ref{eq:estimate_belief}). We instantiate the rollout method with $M=50$ particles, lookahead horizon $\ell=1$, rollout horizon $m=5$, and $L=10$ simulations.}\label{fig:cost5}
\end{figure}

Figure \ref{fig:cost5} shows a comparison between the cost of the base policy $\mu$ and the rollout policy $\tilde{\mu}$ for increasing numbers of replicas $N$ in scenario $1$. We observe that singleagent rollout and multiagent rollout consistently improve upon the base policy across all system sizes, as expected from Prop.~\ref{prop:improvement} and Prop.~\ref{prop:multiagent_rollout_improvement}. For instance, with $N=7$ replicas, the base policy incurs a cost of $56.8$, whereas singleagent rollout reduces this cost to $38.2$, corresponding to an improvement of about $33\%$. Similar experimental results for singleagent rollout were observed by Liu et al. \cite{LIU2022110473}. Likewise, for $N=30$, the base policy incurs a cost of $139.3$, whereas multiagent rollout reduces this cost to $127.8$, corresponding to an improvement of about 9\%. We also note that autonomous multiagent rollout with both of the signaling policies described in Section~\ref{sec:experiment_setup} achieves a lower cost than the base policy, but a slightly higher cost than multiagent rollout. The difference in cost achieved by autonomous multiagent rollout when changing the signaling policy from the base policy $\widehat{\mu}$ to the neural network signaling policy $\widehat{\mu}_{\theta}$ is relatively small. This indicates that the need for coordination among the agents is rather low in this scenario.

\begin{figure}
  \centering
\scalebox{0.65}{
    \includegraphics{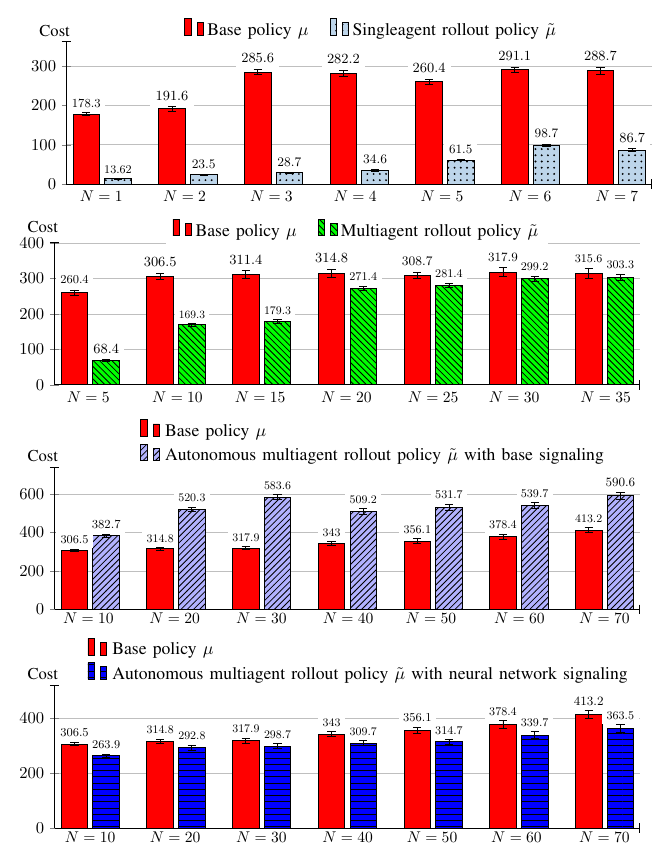}    
  }
  \caption{Simulation scenario 2: Comparison between the cost of the base policy $\mu$ [cf.~Eq.~\eqref{eq:base_policy}] and different rollout policies $\tilde{\mu}$. For all policies, we use particle filtering to implement the belief estimator; cf.~Eq.~(\ref{eq:estimate_belief}). We instantiate the rollout method with $M=50$ particles, lookahead horizon $\ell=1$, rollout horizon $m=5$, and $L=10$ simulations.}\label{fig:cost8}
\end{figure}

Figure \ref{fig:cost8} shows a comparison between the cost of the base policy $\mu$ and the rollout policy $\tilde{\mu}$ for increasing numbers of replicas $N$ in scenario $2$. Similar to the results of scenario $1$, we observe that singleagent rollout and multiagent rollout consistently improve upon the base policy across all system sizes. For instance, with $N=7$ replicas, the base policy incurs a cost of $288.7$, whereas singleagent rollout reduces this cost to $86.7$, corresponding to an improvement of about $70\%$. Similarly, for $N=15$, the base policy incurs a cost of $311.4$, whereas multiagent rollout reduces this cost to $179.3$, corresponding to an improvement of about $42$\%. 

However, in contrast to scenario $1$, we note that the cost incurred by autonomous multiagent rollout with the base policy as signaling policy leads to a \textit{higher} cost than the base policy. For example, with $N=70$ replicas, the base policy incurs a cost of $413.2$, whereas autonomous multiagent rollout with the base policy as signaling policy increases this cost to $590.6$, corresponding to an increase of about $43$\%. On the other hand, when using the neural network signaling policy, autonomous multiagent rollout incurs a lower cost than the base policy. Similar experimental results were observed by Bhattacharya et al. when applying the autonomous multiagent rollout method to a robot-repair problem \cite{pmlr-v155-bhattacharya21a}.

Finally, Fig.~\ref{fig:signaling1} compares the performance of the autonomous multiagent rollout policy [cf.~Eq.~\eqref{eq:auto_multiagent_rollout}] with the neural network signaling policy. We observe that both the rollout and signaling policies outperform the base policy, with the rollout policy achieving the best overall performance. 

\begin{figure}
  \centering
\scalebox{0.65}{
    \includegraphics{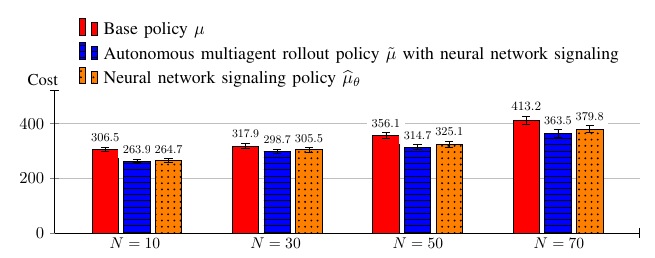}    
  }
  \caption{Simulation scenario 2: Comparison between the cost of the base policy $\mu$ [cf.~Eq.~\eqref{eq:base_policy}], the autonomous multiagent rollout policy $\tilde{\mu}$ [cf.~Eq.~\eqref{eq:auto_multiagent_rollout}], and the neural network signaling policy $\widehat{\mu}_{\theta}$. For all policies, we use particle filtering to implement the belief estimator; cf.~Eq.~(\ref{eq:estimate_belief}). We instantiate the rollout method with $M=50$ particles, lookahead horizon $\ell=1$, rollout horizon $m=5$, and $L=10$ simulations.}\label{fig:signaling1}
\end{figure}

\vspace{2mm}
\noindent\textit{\textbf{Testbed evaluation.}} To evaluate the potential benefit of our rollout method compared to the periodic recovery control policies used in practice, we run our method on a testbed where we execute real service replicas and failure events. For details about our testbed setup, see Appendix \ref{app:testbed}. 

The testbed evaluation is different from the simulation studies in the following ways. First, in the simulation studies, the recovery controls are artificial in the sense that they have no effect on an operational system. By contrast, in the testbed evaluation, each control $u^i$ corresponds to the recovery of an operational service replica. Second, the observations in the simulations are sampled from the distribution in Fig.~\ref{fig:obs_dist}. By contrast, in the testbed evaluation, each observation $z^i$ is measured from logs generated by an operational system. As a consequence, the evaluation on the testbed takes much longer than the simulation-based evaluations. In particular, each time step on the testbed takes around $30$ seconds to complete the recovery and allow the effects to propagate through the system. In our evaluation, we run more than $5000$ time steps to evaluate the performance of different control policies under different conditions. Consequently, the evaluation takes about two days of continuous execution on our testbed.

We run two testbed scenarios:
\begin{itemize}
\item \textit{Testbed scenario 1}: In this scenario, we run a system with $N=10$ service replicas on our testbed. Due to the relatively small number of replicas, we can apply both multiagent rollout and autonomous multiagent rollout.
\item \textit{Testbed scenario 2}: In this scenario, we run a system with $N=50$ service replicas on our testbed. Due to the large number of replicas, the only practical rollout method is the autonomous multiagent rollout method.
\end{itemize}
\begin{figure}
  \centering
\scalebox{0.65}{
    \includegraphics{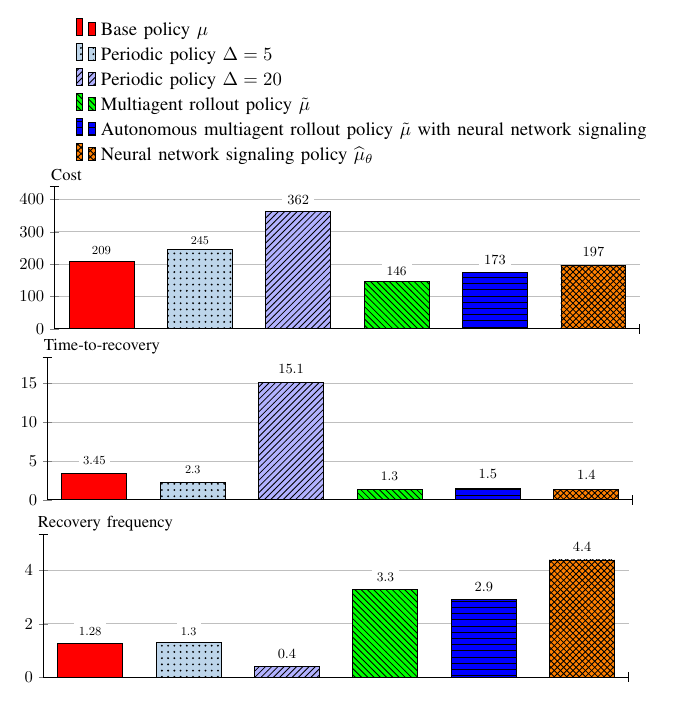}    
  }
  \caption{Testbed scenario $1$: a system with $N=10$ service replicas. The figure shows a performance comparison between multiagent rollout policies [cf.~Eqs.~\eqref{eq:multiagent_rollout}--\eqref{eq:auto_multiagent_rollout}], the base policy [cf.~Eq.~\eqref{eq:base_policy}], the neural network signaling policy $\widehat{\mu}_{\theta}$, and periodic recovery policies with period $\Delta$. For all policies, we use particle filtering to implement the belief estimator; cf.~Eq.~(\ref{eq:estimate_belief}). We instantiate the rollout method with $M=50$ particles, lookahead horizon $\ell=1$, rollout horizon $m=5$, and $L=10$ simulations.}\label{fig:testbed}
\end{figure}

In both scenarios, we configure the parameters of the POMDP in the same way as in Simulation scenario $2$, except for the failure-dependency matrix $A$, which depends on the testbed configuration.\footnote{Our testbed supports $10$ different software versions for the replicas; the versions are listed in \cite[Table 4.5]{kim_phd_thesis}. When we start a replica in the testbed, we select the version uniformly at random. If replicas $i$ and $j$ share the same software version, they are vulnerable to the same types of failures. We model this dependency by setting $A_{ij} = A_{ji} = 1$ if replicas $i$ and $j$ run the same version, and $A_{ij}=A_{ji} = 0$ otherwise.} We compare the performance of the rollout policies with the base policy $\mu$ [cf.~Eq.~\eqref{eq:base_policy}] and two periodic recovery policies. The first periodic policy recovers a replica every $\Delta=5$ time steps, and the second policy recovers a replica every $\Delta=20$ time steps. (We introduce an offset between replicas so that they do not recover simultaneously.)

We consider three evaluation metrics: the cost [cf.~Eq.~\eqref{eq:stage_cost}], the time-to-recovery (i.e., the average time from failure to recovery of a replica), and the recovery frequency (i.e., how many replicas are recovered per time step on average).

The evaluation results are shown in Figs. \ref{fig:testbed}--\ref{fig:testbed2}. We observe that multiagent rollout significantly outperforms both the base policy and the periodic policies across all metrics. For example, the average time-to-recovery of multiagent rollout in Testbed scenario $1$ is $1.3$, whereas the average time-to-recovery of the periodic policy with recovery period $\Delta=20$ is $15.1$. Similarly, the cost of autonomous multiagent rollout is $173$, whereas the cost of the base policy is $209$. 

Moreover, we note that multiagent rollout incurs a slightly lower cost than that of the autonomous version in scenario $1$ ($146$ versus $173$). However, the multiagent rollout method is only applicable to systems with a relatively small number of service replicas. In particular, we were unable to evaluate the multiagent rollout policy in Testbed scenario $2$ due to computational reasons. By contrast, the autonomous multiagent rollout method scales favorably to systems with many replicas.

When comparing the autonomous multiagent rollout policy with the neural network signaling policy and the base policy, we observe that autonomous multiagent rollout achieves the lowest cost in both scenarios. However, the difference in cost between the neural network signaling policy and the autonomous multiagent rollout policy is relatively low.

\begin{figure}
  \centering
\scalebox{0.65}{
    \includegraphics{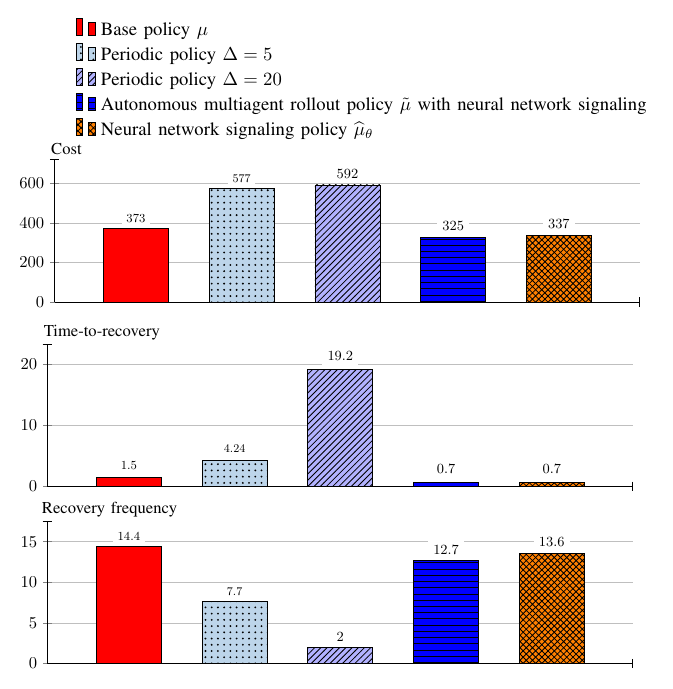}    
  }
  \caption{Testbed scenario $2$: a system with $N=50$ service replicas. The figure shows a performance comparison between the autonomous multiagent rollout policy [cf.~Eq.~\eqref{eq:auto_multiagent_rollout}], the base policy [cf.~Eq.~\eqref{eq:base_policy}], the neural network signaling policy $\widehat{\mu}_{\theta}$, and periodic recovery policies with period $\Delta \in \{5, 20\}$. For all policies, we use particle filtering to implement the belief estimator; cf.~Eq.~(\ref{eq:estimate_belief}). We instantiate the rollout method with $M=50$ particles, lookahead horizon $\ell=1$, rollout horizon $m=5$, and $L=10$ simulations.}\label{fig:testbed2}
\end{figure}
\section{Discussion of the Evaluation Results}
In summary, our main experimental findings are:
\begin{itemize}
\item \textit{Multiagent rollout is more scalable than singleagent rollout for systems with combinatorial control spaces.}
\begin{itemize}
\item In our setup, singleagent rollout becomes computationally intractable for systems with more than a few replicas ($N>8$) due to the combinatorial explosion of the control space. In contrast, multiagent rollout scales near-linearly, and autonomous multiagent rollout achieves nearly constant compute time as the number of replicas grows; see~Fig.~\ref{fig:scale}.
\end{itemize}
\item \textit{Rollout policies consistently outperform the base policy.}
\begin{itemize}
\item Across all system sizes and evaluation scenarios, the rollout policy with a suitable signaling policy leads to a significant cost reduction compared to the base policy. A theoretical explanation for the large cost reduction is provided in the textbook \cite{bertsekas2022lessons} by Bertsekas, where it is shown that the rollout computation performs one step of Newton’s method for solving Bellman’s equation. As a consequence, if the base policy is sufficiently close to optimal, the rollout policy can exhibit a superlinear rate of convergence to the optimal, just like Newton’s method.
\end{itemize}
\item \textit{Autonomous multiagent rollout with a learned signaling policy is nearly as effective as multiagent rollout.}
\begin{itemize}
\item While the multiagent rollout policy performs better than the autonomous version in all evaluation scenarios, the performance difference is relatively small; cf.~Figs.~\ref{fig:cost5}--\ref{fig:testbed}. In particular, when taking both scalability and performance into account, the autonomous multiagent rollout method stands out as the most suitable for large-scale systems; see~Fig.~\ref{fig:scale}.
\end{itemize}    
\item \textit{Autonomous multiagent rollout significantly outperforms the periodic recovery policies used in practice.}
\begin{itemize}
\item The testbed evaluation shows that our method reduces costs by up to 45\% compared to the periodic policies currently used in practice; cf.~Fig.~\ref{fig:testbed2}.
\end{itemize} 
\end{itemize}    
\section{Conclusion}
We addressed the problem of recovery control in replicated computing systems. We formulated this problem as a POMDP with a multiagent structure. We exploited this structure to apply a multiagent rollout method for approximating optimal recovery control policies. Through this formulation, our method overcomes the combinatorial complexity that renders standard rollout methods intractable for systems with more than a few replicas. Moreover, by using precomputed signaling information, our method reduces the need for replica coordination and facilitates parallel computations. We evaluated our system through extensive simulations and testbed experiments. The evaluation results confirm that our method scales to systems with a large number of replicas and significantly reduces operational costs compared to the heuristic recovery policies commonly used in operational systems. Furthermore, we characterized the computational complexity of our method and defined theoretical conditions under which it has an approximate cost improvement property.

\section*{Acknowledgment}
The authors are deeply grateful to Prof. Dimitri Bertsekas, who passed away on June 3, 2026, for his insightful comments and encouragement during the development of this paper.

\appendices

\section{Dynamic Programming Preliminaries}\label{app:prelim}
Before presenting the proofs of Props.~\ref{prop:improvement} and \ref{prop:multiagent_rollout_improvement}, we state some standard definitions and results from dynamic programming that will be used repeatedly in the analysis. In our terminology, all equations, inequalities, and convergence limits involving functions are meant to be pointwise.

Let $T$ and $T_{\mu}$ be Bellman operators defined as 
\begin{align*}
(TJ)(b) &= \min_{u \in U}\left[\hat{g}(b, u) + \alpha \sum_{z \in Z}\hat{p}(z \mid b, u)J(F(b,u,z))\right],\\
(T_{\mu}J)(b) &= \hat{g}(b, \mu(b)) + \alpha \sum_{z \in Z}\hat{p}(z \mid b, \mu(b))J(F(b,\mu(b),z)),
\end{align*}
for all beliefs $b \in B$. As is well known, these operators are monotonic, i.e., if $J\leq J'$, then $TJ\leq TJ'$ and $T_{\mu}J\leq T_{\mu}J'$; see, e.g., \cite[Lem.~1.1.1]{bertsekas2012dynamic}. Moreover, they have the constant shift property, i.e., $T^{\ell}(J+\epsilon e) = T^{\ell}J + \alpha^{\ell}\epsilon e$ and $T^{\ell}_{\mu}(J+\epsilon e) = T^{\ell}_{\mu}J + \alpha^{\ell}\epsilon e$, where $e$ is the unit function i.e., $e(b)\equiv1$, and $\ell \geq 1$ is an integer; see e.g., \cite[Lem.~1.1.2]{bertsekas2012dynamic}. Furthermore, our proof relies on the following classical results; cf. \cite[Prop.~1.2.4]{bertsekas2012dynamic}.
\begin{proposition}\label{prop:bert_upper}
For every policy $\mu$, we have the following:
\begin{itemize}
    \item[(a)] Its associated cost function $J_\mu$ satisfies $J_\mu=T_\mu J_\mu$.
    \item[(b)] If a bounded function $J$ satisfies $T_{\mu}J \leq J$, then $J_{\mu} \leq J$.
\end{itemize}
\end{proposition}
The proof of part (b) was originally given by Bertsekas and can be found in \cite[Prop. 1b]{doi:10.1137/0315031}.
\section{Proof of Proposition \ref{prop:improvement}}\label{app:proof_prop_1}
We start by establishing the existence of a constant $c$ such that inequality \eqref{eq:lyapunov_rollout} holds. Since the functions $\tilde J_\mu$ and $\hat{g}$ are bounded [cf.~Eqs.~\eqref{eq:stage_cost} and \eqref{eq:rollout_approximation}], we can choose a constant $c$ sufficiently large to satisfy the inequality. As a result, $\epsilon < \infty$.

Assume that the infimum in Eq.~\eqref{eq:epsilon_approx_rollout} is attained at $\hat c$. Consider the function $\hat{J}_\mu=\tilde J_\mu+\hat c e$. Inequality \eqref{eq:lyapunov_rollout} states that $T\hat J_\mu\leq \hat J_\mu$. By monotonicity of $T$, we have $T^{\ell}\hat J_\mu\leq T^{\ell-1}\hat J_\mu$. Similarly, applying $T_{\tilde \mu}$ to both sides, we have 
\begin{align}
T_{\tilde \mu}(T^{\ell}\hat J_\mu)\leq T_{\tilde \mu}(T^{\ell-1}\hat J_\mu). \label{proof_prop1_step1}
\end{align}
Equation~\eqref{eq:rollout} implies that $T_{\tilde \mu}(T^{\ell-1}\tilde J_\mu)=T^{\ell}\tilde J_\mu$, which means that $T_{\tilde \mu}(T^{\ell-1}\hat J_\mu)=T^{\ell}\hat J_\mu$. Combining this equality with inequality \eqref{proof_prop1_step1}, we obtain $T_{\tilde \mu}(T^{\ell}\hat J_\mu)\leq T^{\ell}\hat J_\mu$. This inequality means that $J_{\tilde \mu}\leq T^\ell\hat J_\mu$; cf.~Prop.~\ref{prop:bert_upper}b. In view of Eq.~\eqref{eq:epsilon_approx_rollout}, we have $T^{\ell}\hat J_\mu\leq T^{\ell} (J_\mu+\epsilon e) \leq T^{\ell}J_\mu+\alpha^{\ell}\epsilon e$, where the last inequality follows from the constant shift property. From the definitions $T_\mu$ and $T$, we have $TJ_{\mu}\leq T_{\mu}J_{\mu}=J_{\mu}$, which implies that $T^{\ell}J_\mu \leq T_{\mu}^{\ell}J_\mu =J_{\mu}$, where the equalities are due to Prop~\ref{prop:bert_upper}a. Consequently, $T^{\ell}J_\mu+\alpha^{\ell}\epsilon e\leq J_\mu+\alpha^{\ell}\epsilon e$. As a result, we obtain the desired inequality $J_{\tilde \mu}\leq J_\mu+\alpha^{\ell}\epsilon e$.

Lastly, let us consider the case where the infimum of Eq.~\eqref{eq:epsilon_approx_rollout} is not attained. In this case, there exists a sequence $\{\hat c_k\}$ such that each $\hat c_k$ satisfies the constraint \eqref{eq:lyapunov_rollout}, and $\epsilon_k$ converges to $\epsilon$, where $\epsilon_k=\Vert \tilde J_\mu+\hat c_ke -J_\mu\Vert$. Using the preceding arguments with $\hat c_k$ in place of $\hat c$, and $\epsilon_k$ in place of $\epsilon$, we can show that $J_{\tilde \mu}\leq J_\mu+\alpha^{\ell}\epsilon_k e$ for all $k$. Taking the limit on the right-hand side, we get the desired inequality. \qed

\section{Proof of Proposition \ref{prop:multiagent_rollout_improvement}}\label{app:proof_prop_3}
We start by showing that $\hat\epsilon<\infty$; cf.~Eq.~\eqref{eq:epsilon_approx_multi_rollout}. To see this, consider $c=\Vert T_{\tilde\mu}\tilde J_\mu-\tilde J_\mu\Vert/(1-\alpha)$. Clearly, the value $c$ satisfies the constraint \eqref{eq:lyapunov_multi_rollout}, which implies that $\hat\epsilon<\infty$.

Suppose that the infimum in Eq.~\eqref{eq:epsilon_approx_multi_rollout} is attained at $\hat c$. Consider the function $\hat J_\mu=\tilde J_\mu+\hat c e$. Inequality \eqref{eq:lyapunov_multi_rollout} states that $T_{\tilde\mu}\hat J_\mu\leq \hat J_\mu$, which implies $J_{\tilde \mu}\leq \hat J_\mu$; cf.~Prop.~\ref{prop:bert_upper}.b. In addition, Eq.~\eqref{eq:epsilon_approx_multi_rollout} implies that $\hat J_\mu\leq J_\mu+\hat\epsilon e$. Combining this inequality with $J_{\tilde \mu}\leq \hat J_\mu$, we obtain the desired result $J_{\tilde \mu}\leq J_\mu+\hat\epsilon e$.
        
The case where the infimum $\hat\epsilon$ is not attained follows similarly to that of Prop.~\ref{prop:improvement} and is thus omitted. \qed

\section{Notation}\label{app:notation}
Our notation is summarized in Table \ref{tab:notation}.

\begin{table}
  \centering
  \scalebox{0.7}{
    \begin{tabular}{ll} \toprule
\rowcolor{lightgray}
      {\textit{Notation(s)}} & {\textit{Description}} \\ \midrule
      $X,Z,U,B$ & State, observation, control, and belief spaces; cf.~Section~\ref{sec:recovery_pomdp}.  \\
      $N$ & Number of replicas; cf.~Section~\ref{sec:system_model}.  \\
      $f$ & Tolerance threshold; cf.~Eq.~\eqref{eq:downtime_cost}.  \\      
      $i$ & Replica index; cf.~Section~\ref{sec:system_model}.\\            
      $M$ & Number of particles; cf.~Eq.~\eqref{eq:estimate_belief}  \\
      $\hat{b}$ & Belief estimated via particle filtering; cf.~Eq.~\eqref{eq:estimate_belief}  \\      
      $w$ & Maximum number of failure alerts per replica; cf.~Fig.~\ref{fig:spaces}.  \\           
      $F$ & Belief estimator; cf.~Section~\ref{sec:recovery_pomdp}.  \\            
      $x,y$ & State values; cf.~Section~\ref{sec:recovery_pomdp}.  \\
      $p_{xy}(u)$ & Probability of the state transition $x\rightarrow y$ under control $u$; cf.~Section~\ref{sec:recovery_pomdp}.  \\
      $p(z \mid y,u)$ & Probability of the observation $z$ given state $y$ and control $u$; cf.~Section~\ref{sec:recovery_pomdp}.  \\
      $g(x,u,y)$ & Stage cost function; cf.~Section~\ref{sec:recovery_pomdp}.  \\
      $\hat{g}(b,u)$ & Expected stage cost given belief $b$ and control $u$; cf.~Section~\ref{sec:recovery_pomdp}. \\
      $X^i$ & Set of states where $x^i=1$; cf.~Section~\ref{sec:recovery_pomdp}. \\
      $b^i$ & Marginal belief that $x^i=1$; cf.~Eq.~\ref{eq:marginal_b}. \\            
      $\hat{p}(z \mid b, u)$ & Probability of observation $z$ given belief $b$ and control $u$; cf.~Section~\ref{sec:recovery_pomdp}. \\
      $p_{\mathrm{F}}$ & Factor for the failure probability; cf.~Eq.~\eqref{eq:failure_prob}.\\
      $A$ & Failure-dependency matrix; cf.~Eq.~\eqref{eq:failure_prob}.\\
      $\lambda$ & Factor for the service disruption cost; cf.~Eq.~\eqref{eq:stage_cost}.\\
      $\eta$ & Factor for the failure cost; cf.~Eq.~\eqref{eq:stage_cost}.\\
      $\phi(x,u)$ & Function determining the service disruption cost; cf.~Eq.~\eqref{eq:downtime_cost}.\\                                          
      $\alpha$ & Discount factor; cf.~Section~\ref{sec:recovery_pomdp}.  \\            
      $k$ & Time step; cf.~Section~\ref{sec:recovery_pomdp}.  \\
      $x^i,b^i,u^i,z^i$ & State, belief, control, and observation related to replica $i$; cf.~Section~\ref{sec:recovery_pomdp}.\\      
      $x_k,b_k,u_k,z_k$ & State, belief, control, and observation at time step $k$; cf.~Section~\ref{sec:recovery_pomdp}. \\
      $J^{*}$ & Optimal cost function; cf.~Section~\ref{sec:recovery_pomdp}.  \\
      $J_{\mu}$ & Cost function of policy $\mu$; cf.~Section~\ref{sec:recovery_pomdp}.  \\
      $\mu$ & Base policy; cf.~Eq.~\eqref{eq:rollout}.  \\
      $\tilde{\mu}$ & Rollout policy and its variants; cf.~Eq.~\eqref{eq:rollout}. \\      
      $\mu^i$ & Control component prescribed for replica $i$ by policy $\mu$. \\      
      $L$ & Number of rollout simulations; cf.~Eq.~\eqref{eq:rollout}.  \\
      $m$ & Rollout horizon; cf.~Eq.~\eqref{eq:rollout}.  \\
      $\ell$ & Lookahead horizon; cf.~Eq.~\eqref{eq:rollout}.  \\
      $\Tilde{J}$ & Terminal cost function approximation; cf.~Eq.~\eqref{eq:rollout_approximation}.  \\
      $\Tilde{J}_{\mu}$ & Estimated cost function of the base policy $\mu$; cf.~Eq.~\eqref{eq:rollout_approximation}. \\      
      $\widehat{\mu}$ & Base signaling policy; cf.~Section~\ref{sec:multiagent_rollout}.  \\
      $\widehat{\mu}_{\theta}$ & Neural network signaling policy; cf.~Section~\ref{sec:experiment_setup}.\\
      $\theta$ & Parameters of a neural network; cf.~Appendix~\ref{app:neural_training}.\\
      $\gamma$ & Learning rate for neural network training; cf.~Appendix~\ref{app:neural_training}. \\
      $\mathcal{M}$ & Minibatch for neural network training; cf.~Appendix~\ref{app:neural_training}. \\
      $r$ & Minibatch size for neural network training; cf.~Appendix~\ref{app:neural_training}. \\      
      $\mathcal{L}$ & Loss function for neural network training; cf.~Appendix~\ref{app:neural_training}. \\
      $\hat{u}^i$ & Control predicted for replica $i$ by the neural network; cf.~Appendix~\ref{app:neural_training}.\\
      $\tilde{u}_j$ & $j$th example of multiagent control for neural network training; cf.~Appendix~\ref{app:neural_training}.\\
      $q$ & Size of training dataset; cf.~Appendix~\ref{app:neural_training}.\\
      $\norm{\cdot}$ & Maximum norm; cf.~Prop.~\ref{prop:improvement}.\\
      $E\{\cdot\}$ & Expectation operator; cf.~Section~\ref{sec:recovery_pomdp}.\\
      $\delta_{ij}$ & Kronecker's delta; cf.~Eq.~\eqref{eq:estimate_belief}.\\
      \bottomrule\\
  \end{tabular}}
  \caption{Notation.}\label{tab:notation}
\end{table}

\section{Neural Network Training}\label{app:neural_training}
To obtain the signaling policy $\widehat{\mu}_{\theta}$, we implement a neural network with $4$ hidden layers, each containing $64$ neurons. The network parameters are denoted by $\theta \in \Re^{d}$. All hidden layers use the ReLU activation function, while the output layer consists of $N$ neurons with Sigmoid activations. Thus, the neural network outputs a vector in $\Re^{N}$ where the $i$th component is a value in $[0,1]$ and can be interpreted as the probability of applying control $1$ to replica $i$.

Since the belief space $B$ has dimension $2^N$ in our model, it is computationally infeasible to have the input layer of the neural network match the dimension of $B$. For this reason, we define the input layer to have $N$ neurons instead of $2^N$. To encode a belief vector $b$ as input, we map it to the vector of marginal failure probabilities under belief $b$ [cf.~Eq.~\eqref{eq:marginal_b}], i.e.,
\begin{align*}
(b^1, b^2, \hdots, b^N).
\end{align*}
To train the parameters of the neural network, we start by generating a training dataset of belief-control pairs. We do so by simulating the POMDP and applying the controls prescribed by the multiagent rollout policy; cf.~Eq.~\eqref{eq:multiagent_rollout}. Such simulations allow us to collect a dataset of the form 
\begin{align}
(b_j, \tilde{u}_j), && j=1,\hdots,q, \label{eq:dataset}
\end{align}
where $b_j$ is the $j$th belief state in the dataset, $\tilde{u}_j$ is the corresponding control prescribed by multiagent rollout, and $q$ is the size of the dataset. For the experiments reported in this paper, we use a dataset of size $q=10^5$.

Given the dataset in Eq.~\eqref{eq:dataset}, we train the parameters $\theta$ via incremental gradient descent to minimize the component-wise cross-entropy loss function
\begin{align}
\mathcal{L}_j(\theta) &= -\frac{1}{N}\sum_{i=1}^N\Big(\tilde{u}_j^i\ln \hat{u}_j^i + (1-\tilde{u}_j^i)\ln(1-\hat{u}_j^i)\Big),\label{eq:loss}
\end{align}
where $\mathcal{L}_j(\theta)$ is the loss for the belief-control pair $(b_j, \tilde{u}_j)$ in the training dataset and $\hat{u}_j$ is the output of the neural network.

Using the loss function in Eq.~\eqref{eq:loss}, we iteratively update the parameters $\theta$ according to
\begin{align}\label{eq:inc_gradient}
\theta^{t+1} &= \theta^t - \gamma\sum_{j\in \mathcal{M}}\nabla \mathcal{L}_j(\theta^t), & t=1,2,\hdots,
\end{align}
where $t$ is the training iteration, $\gamma$ is the step size, and $\mathcal{M}$ is a set of $r < q$ indices of sampled belief-control pairs from the dataset in Eq.~\eqref{eq:dataset}. For the experiments reported in this paper, we use batch size $r=512$ and learning rate $\gamma=0.0001$.

After training the parameters $\theta$, we use the neural network to implement the signaling policy $\widehat{\mu}_{\theta}$ as follows
\begin{align}\label{eq:imitation_signaling}
  \widehat{\mu}^i_{\theta}(b) &=
                       \begin{dcases}
                         1 & \text{ if }\hat{u}^i > 0.5,\\
                         0 & \text{otherwise},
                       \end{dcases}  
\end{align}
where $\hat{u}$ is the probability vector $\hat{u} \in [0,1]^N$ that is output by the neural network. Figure \ref{fig:imitation} shows the loss and accuracy of the neural network during training.

\begin{figure}[H]
  \centering
\scalebox{0.65}{
    \includegraphics{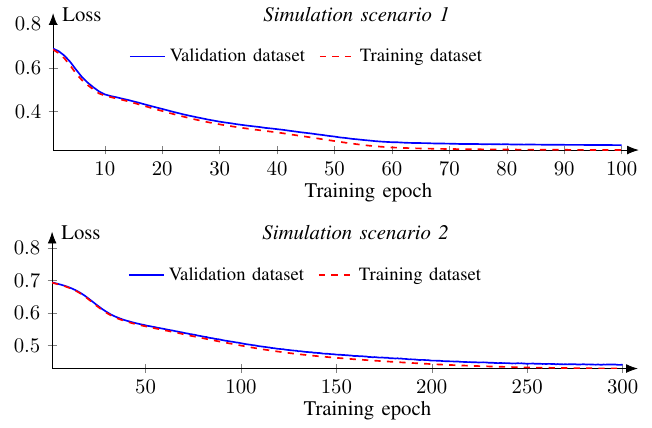}
  }
\caption{Loss of the signaling policy $\widehat{\mu}_{\theta}$ on the training and validation datasets with $N=10$; cf.~Eq.~\eqref{eq:loss}. The dataset in Eq.~\eqref{eq:dataset} is divided into two parts: $80\%$ is used for training and $20\%$ is reserved for validation. A training epoch refers to one complete pass of the incremental gradient algorithm [cf.~Eq.~\eqref{eq:inc_gradient}] through the entire training dataset.}\label{fig:imitation}
\end{figure}

\section{Auxiliary Computational Results}\label{app:extra_results}
In this appendix, we provide additional experimental results that provide insights into computational aspects of our method but are not central to the main research questions addressed in the paper. In particular, these results illustrate the impact of various configuration parameters of our method.

Figure \ref{fig:scale2} shows the compute time of the singleagent rollout method for varying configurations of the lookahead horizon $\ell$, the rollout horizon $m$, the number of simulations $L$, and the number of belief particles $M$. As expected, increasing any of these parameters results in additional compute time. However, we observe that increasing $m,L,M$ has a low impact on the compute time. In particular, these parameters can be scaled up to $10^3$ while still keeping the total compute time below 5 seconds for $N \leq 3$. By contrast, the compute time grows exponentially when increasing the lookahead horizon $\ell$.

\begin{figure}
  \centering
\scalebox{0.65}{
    \includegraphics{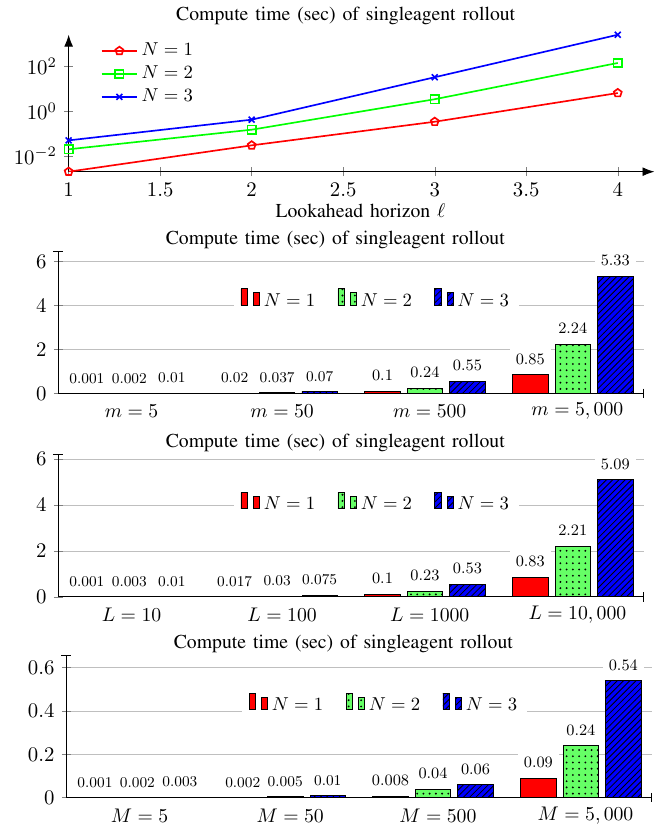}    
  }
  \caption{Compute time of rollout with different configurations of the lookahead horizon $\ell$, the rollout horizon $m$, the number of simulations $L$, and the number of particles $M$ for the singleagent rollout method; cf.~Eqs.~\eqref{eq:rollout} and \eqref{eq:estimate_belief}. For all policies, we use particle filtering to implement the belief estimator; cf.~Eq.~(\ref{eq:estimate_belief}). Unless indicated otherwise, we use lookahead horizon $\ell=1$, rollout horizon $m=5$, $L=10$ simulations, and $M=50$ particles.}\label{fig:scale2}
\end{figure}

Lastly, Fig. \ref{fig:ablation} presents a comparison between singleagent and multiagent rollout policies with different belief estimators. The results show that while all rollout policies offer a substantial improvement over the base policy, the difference in performance between singleagent rollout and multiagent rollout is marginal. Moreover, the performance difference when changing the belief estimator is relatively low.

\begin{figure}
  \centering
\scalebox{0.65}{
    \includegraphics{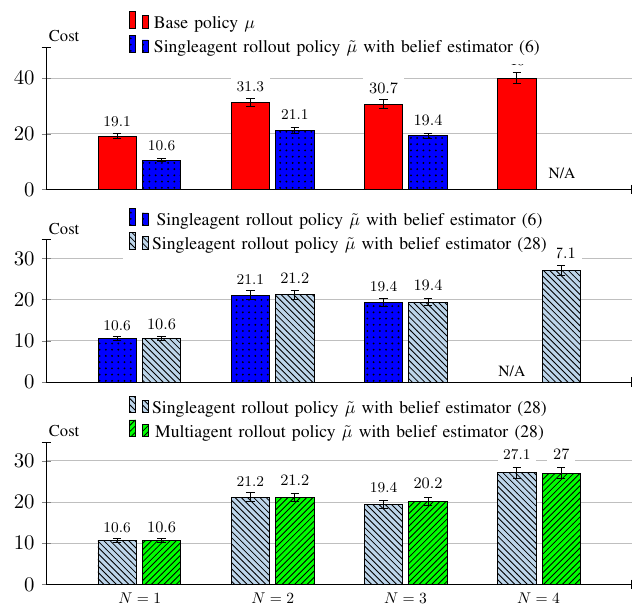}    
  }
\caption{Simulation scenario $1$: cost of the base policy $\mu$ and rollout policies with different belief estimators. Bar groups indicate the number of replicas $N$. Numbers and error bars indicate the mean and the standard deviation from $100$ simulations with different random seeds. Each simulation includes $100$ time steps. All rollout methods are instantiated with $M=50$ particles, lookahead horizon $\ell=1$, rollout horizon $m=5$, and $L=10$ simulations. The bar colors relate to different control policies. A missing bar and the text ``N/A'' indicates that a result could not be obtained within $2$ hours of computation.}\label{fig:ablation}
\end{figure}

\section{Experimental Testbed}\label{app:testbed}
We implement the replicated computing system in Fig.~\ref{fig:example} as a proof-of-concept on our testbed. The implementation includes three layers, as described below.

\vspace{2mm}
\noindent\textit{\textbf{The physical layer.}}
The physical layer contains a cluster with $13$ nodes connected through an Ethernet network; see Fig.~\ref{fig:dt_serverrack}. Specifications of the nodes can be found in \cite[Table 4.4]{kim_phd_thesis}. Nodes communicate via message passing over authenticated channels. Each node runs (\textit{i}) a service replica in a Docker container \cite{docker}; (\textit{ii}) a control agent $\mu^i$; and (\textit{iii}) the Snort intrusion detection system with ruleset v2.9.17.1 \cite{snort}. The container images are available in  \cite[Table 4.5]{kim_phd_thesis}.

Network connectivity between replicas is emulated with virtual links implemented by Linux bridges and network namespaces, which create copies of the physical host's network stack. If an emulated network spans multiple physical servers, the traffic is tunneled over the physical network using VXLAN. In other words, the physical network provides a substrate, on top of which the emulated networks are overlaid. Replicas are interconnected through Gbit/s connections with $0.05\%$ packet loss (emulated with NetEm). They are connected to clients via $100$ Mbit/s connections with $0.1\%$ packet loss.

\vspace{2mm}
\noindent\textit{\textbf{The virtualization layer.}} Each service replica is a state machine and runs a web service; see Schneider for a tutorial on state machines \cite{state_machine_reference}. This service offers two deterministic operations: (\textit{i}) a \textit{read operation}, which returns the current state of the service; and (\textit{ii}) a \textit{write operation}, which updates the state. Replicas run the MINBFT consensus protocol to coordinate these operations; see \cite{6081855} for details about MINBFT. The throughput of our implementation of MINBFT is analyzed in \cite[Fig. 4.13]{kim_phd_thesis}. Clients access the service by issuing requests that are sent to all replicas. Each request has a unique identifier that is digitally signed. After sending a request, the client waits for a quorum of $f+1$ identical replies with valid signatures. 

The client population is emulated through processes that send service requests to the replicas. Client arrivals are controlled by a Poisson process with exponentially distributed service times. The sequence of service invocations is selected according to a Markov process; see \cite{kim_phd_thesis} for details.

We emulate two types of failure events at each time step: crashes and cyberattacks. The attacks are listed in \cite[Tables 4.5-4.6]{kim_phd_thesis}. A crashed replica does not participate in the consensus protocol until it is recovered. After compromising a replica through an emulated attack, we randomly choose between a) having the replica participate in the consensus protocol; b) not participating; and c) participating with random messages.

\vspace{2mm}
\noindent\textit{\textbf{The control layer.}}
Control agents collect failure alerts from Snort \cite{snort} and decide when to recover replicas. When a replica is recovered, it starts with a new container, and its state is initialized with the (identical) state from $f+1$ other replicas.

\bibliographystyle{IEEEtran}
\bibliography{references}

\end{document}